\documentclass[10pt,twocolumn,journal]{IEEEtran}

\usepackage{amsmath, mathrsfs, mathtools}
\usepackage{amsfonts}
\usepackage{amssymb}
\usepackage{color}
\usepackage{multirow}
\usepackage{stmaryrd}

\allowdisplaybreaks

\usepackage{caption}
\usepackage{subcaption}

\usepackage{amsfonts}
\usepackage{amssymb}
\usepackage{color}
\usepackage{multirow}
\usepackage{graphicx}

\newcommand{\subparagraph}{}
\usepackage[compact]{titlesec}
\titlespacing*{\section}{15pt}{1.2\baselineskip}{0.9\baselineskip}

\allowdisplaybreaks

\newcommand{\myhash}{%
  {\settoheight{\dimen0}{C}\kern-.05em\, \resizebox{!}{\dimen0}{\raisebox{\depth}{\#}}}}

\usepackage[numbers,sort&compress]{natbib}

\DeclarePairedDelimiter\floor{\lfloor}{\rfloor}

\def\aul{\bfa_{\text{ul}}}
\def\adl{\bfa_{\text{dl}}}
\def\sigul{{\boldsymbol{\sigma}_{\text{ul}}}}
\def\sigdl{{\boldsymbol{\sigma}_{\text{dl}}}}
\def\Sigul{{\boldsymbol{\Sigma}_{\text{ul}}}}
\def\Sigdl{{\boldsymbol{\Sigma}_{\text{dl}}}}
\def\gamf{{\check{\gamma}}}
\def\lambdam{{\boldsymbol{\lambda}}}
\def\rect{{\ttr\tte\ttc\ttt}}
\def\xim{\boldsymbol{\xi}}
\def\phim{\boldsymbol{\phi}}

\usepackage{multicol}

\usepackage{algorithm}
\usepackage{algpseudocode}
\usepackage{pifont}
\usepackage{varwidth}


\usepackage{mathbbol}

\def\mindex#1{\index{#1}}

\def\sq{\hbox{\rlap{$\sqcap$}$\sqcup$}}
\def\qed{\ifmmode\sq\else{\unskip\nobreak\hfil
\penalty50\hskip1em\null\nobreak\hfil\sq
\parfillskip=0pt\finalhyphendemerits=0\endgraf}\fi\medskip}

\long\def\defbox#1{\framebox[.9\hsize][c]{\parbox{.85\hsize}{%
\parindent=0pt
\baselineskip=12pt plus .1pt      
\parskip=6pt plus 1.5pt minus 1pt 
 #1}}}

\long\def\beginbox#1\endbox{\subsection*{}%
\hbox{\hspace{.05\hsize}\defbox{\medskip#1\bigskip}}%
\subsection*{}}

\def\endbox{}

\def\tr{\mathsf{tr}}

\newsavebox{\junk}
\savebox{\junk}[1.6mm]{\hbox{$|\!|\!|$}}

\def\argmin{\mathop{\rm arg\, min}}

\def\Re{\field{R}}

\def\bC{{\mathbb C}}

\def\bE{{\mathbb E}}

\def\bR{{\mathbb R}}
\def\bS{{\mathbb S}}

\def\bZ{{\mathbb Z}}

\def\bfA{{\bf A}}
\def\bfB{{\bf B}}

\def\bfI{{\bf I}}

\def\bfU{{\bf U}}
\def\bfV{{\bf V}}
\def\bfW{{\bf W}}
\def\bfX{{\bf X}}

\def\bfa{{\bf a}}
\def\bfb{{\bf b}}
\def\bfc{{\bf c}}

\def\bfh{{\bf h}}

\def\bfn{{\bf n}}

\def\bfp{{\bf p}}
\def\bfq{{\bf q}}
\def\bfr{{\bf r}}
\def\bfs{{\bf s}}

\def\bfu{{\bf u}}

\def\bfw{{\bf w}}
\def\bfx{{\bf x}}

\def\ttc{{\mathtt c}}

\def\tte{{\mathtt e}}

\def\ttr{{\mathtt r}}

\def\ttt{{\mathtt t}}

\def\sfH{{\sf H}}

\def\bfmath#1{{\mathchoice{\mbox{\boldmath$#1$}}%
{\mbox{\boldmath$#1$}}%
{\mbox{\boldmath$\scriptstyle#1$}}%
{\mbox{\boldmath$\scriptscriptstyle#1$}}}}

\def\bfmY{\bfmath{Y}}

\def\bfmhhaY{\bfmath{\hhaY}} 
\def\bfmhhaY{\hbox to 0pt{$\widehat{\bfmY}$\hss}\widehat{\phantom{\raise 1.25pt\hbox{$\bfmY$}}}}

\def\til={{\widetilde =}}

\def\clA{{\cal A}}

\def\clC{{\cal C}}
\def\clD{{\cal D}}
\def\clE{{\cal E}}
\def\clF{{\cal F}}
\def\clG{{\cal G}}
\def\clH{{\cal H}}

\def\clI{{\cal I}}

\def\clL{{\cal L}}

\def\clN{{\cal N}}

\def\clR{{\cal R}}

\def\clW{{\cal W}}
\def\clX{{\cal X}}

 \def\FRAC#1#2#3{\genfrac{}{}{}{#1}{#2}{#3}}

\def\ddtp{{\mathchoice{\FRAC{1}{d^{\hbox to 2pt{\rm\tiny +\hss}}}{dt}}%
{\FRAC{1}{d^{\hbox to 2pt{\rm\tiny +\hss}}}{dt}}%
{\FRAC{3}{d^{\hbox to 2pt{\rm\tiny +\hss}}}{dt}}%
{\FRAC{3}{d^{\hbox to 2pt{\rm\tiny +\hss}}}{dt}}}}

\def\average#1,#2,{{1\over #2} \sum_{#1}^{#2}}

\def\eye(#1){{\bf(#1)}\quad}

\newtheorem{theorem}{{\bf Theorem}}
\newtheorem{proposition}{{\bf Proposition}}

\def\eq#1/{(\ref{e:#1})}

\newcommand{\inp}[2]{{\langle #1, #2 \rangle}}

\newcommand{\beqn}[1]{\notes{#1}%
\begin{eqnarray} \elabel{#1}}

\newcommand{\eeqn}{\end{eqnarray} }

\newcommand{\beq}[1]{\notes{#1}%
\begin{equation}\elabel{#1}}

\newcommand{\eeq}{\end{equation}}

\def\bdes{\begin{description}}
\def\edes{\end{description}}

\newcounter{rmnum}

\newcounter{anum}

{\end{list}}

\def\ass(#1:#2){(#1\ref{#1:#2})}

\def\ritem#1{
\item[{\sf \ass(\current_model:#1)}]
}

\newenvironment{recall-ass}[1]{%
\begin{description}
\def\current_model{#1}}{
\end{description}
}

\newtheorem{example}{\bf Example}
\newtheorem{problem}{\bf Problem}

\usepackage{tikz}
\usepackage{pgfplots,tikz-3dplot}
\usetikzlibrary{intersections, pgfplots.fillbetween}
\pgfplotsset{compat=newest}
\usetikzlibrary{patterns}
\usetikzlibrary{positioning}
\usetikzlibrary{datavisualization}
\usetikzlibrary{datavisualization.formats.functions}
\usetikzlibrary{backgrounds}
\usetikzlibrary{shapes,snakes}

\usepgfplotslibrary{external} 
\usepackage{float}
\usepackage{afterpage}
\usepackage{balance}

\def\herm{{\sfH}}

\def\sigmam{{\boldsymbol{\sigma}}}

\def\cg{{\clC\clN}} 
\newcommand{\normd}[1]{{\left\vert\kern-0.25ex\left\vert\kern-0.25ex\left\vert #1 
    \right\vert\kern-0.25ex\right\vert\kern-0.25ex\right\vert}}

\long\def\comment#1{}

\newcommand{\xv}{{\bf x}}

\newcommand{\Xm}{{\bf X}}

\newcommand{\Lambdam}{\hbox{\boldmath$\Lambda$}}

\newcommand{\Sigmam}{\hbox{\boldmath$\Sigma$}}

\renewcommand{\Re}{{\rm Re}}
\renewcommand{\Im}{{\rm Im}}

\newcommand{\transp}{{\sf T}}

\title{Multi-Band Covariance Interpolation with Applications in Massive MIMO}
\author{Saeid Haghighatshoar,  \IEEEmembership{Member, IEEE}, Mahdi Barzegar Khalilsarai,  Giuseppe Caire, \IEEEmembership{Fellow, IEEE} 
\thanks{The authors are with  Communications and Information Theory Group (CommIT), Technische Universit\"{a}t Berlin (\{saeid.haghighatshoar, m.barzegarkhalilsarai, caire\}@tu-berlin.de).}
}

\begin{document}

\maketitle

\begin{abstract}
	In this paper, we study the problem of multi-band (frequency-variant) covariance interpolation with a particular emphasis towards massive MIMO applications. In a massive MIMO system, the communication between each BS with $M \gg 1$ antennas and each \textit{single-antenna} user occurs through a collection of scatterers in the environment, where  the channel vector of each user  at BS antennas consists in a weighted linear combination of the array responses of the scatterers, where each scatterer has its own \textit{angle of arrival} (AoA) and complex channel gain. The array response at a given AoA depends on the wavelength of the incoming planar wave and is naturally frequency dependent. This results in a frequency-dependent distortion where the second order statistics, i.e., the covariance matrix, of the channel vectors  varies with frequency. In this paper, we show that  although this effect is generally negligible for a small number of antennas $M$, it results in a considerable distortion of the covariance matrix and especially its dominant signal subspace in the massive MIMO regime where $M \to \infty$, and can generally incur a serious degradation of the performance especially in \textit{frequency division duplexing} (FDD) massive MIMO systems where the \textit{uplink} (UL) and  the  \textit{downlink} (DL) communication occur  over different frequency bands.
	 We propose a novel UL-DL covariance interpolation technique that is able to recover the  covariance matrix in the DL from an estimate of the covariance matrix in the UL under a mild reciprocity condition on the angular \textit{power spread function} (PSF) of the users (this is in contrast with the UL-DL reciprocity of the instantaneous realization of the channel vectors which does not generally hold). We analyze the performance of our proposed scheme mathematically and prove its robustness under a sufficiently large spatial oversampling of the array. We also propose several \textit{simple off-the-shelf} algorithms for UL-DL covariance interpolation and evaluate their performance via numerical simulations. 
\end{abstract}

\section{Introduction}
Consider a multi-user massive MIMO system \cite{Marzetta-TWC10}, where each Base Station (BS) has an array consisting of $M$ antennas and serves multiple single-antenna users. In this paper, we mainly focus on \textit{Frequency Division Duplexing} (FDD)\footnote{As we will see in the following, the variation of the channel covariance matrix with the frequency also arises in \textit{Time Division Duplexing} (TDD) but is much less significant than that in FDD (as the UL and DL communication is done over the same frequency band)  unless the number of antennas $M$ at the BS is tremendously large.}, where the BS communicates with the users in two disjoint frequency bands $\clF=\clF_{\text{ul}} \cup \clF_{\text{dl}}$, where the users transmit their data to the BS in the \textit{uplink} (UL) over a frequency band  $\clF_{\text{ul}}=[f_\text{ul} - \frac{W_\text{ul}}{2}, f_\text{ul} + \frac{W_\text{ul}}{2}]$ with a carrier frequency $f_\text{ul}$ and a bandwidth $W_\text{ul}$ and receive data from the BS in the \textit{downlink} (DL) in the frequency band $\clF_{\text{dl}}=[f_\text{dl} - \frac{W_\text{dl}}{2}, f_\text{dl} + \frac{W_\text{dl}}{2}]$, where $f_\text{dl}$ denotes the carrier frequency in the DL and where $W_\text{dl}$ is the bandwidth of the DL channel. We always assume that the communication bandwidth is much less than the carrier frequency, i.e., $\frac{W_\text{ul}}{f_\text{ul}}, \frac{W_\text{dl}}{f_\text{dl}} \ll 1$, such that both UL and DL channels can be considered  \textit{quasi} narrow-band. We  define the ratio between the UL and DL carrier frequency by $\nu=\frac{f_\text{ul}}{f_\text{dl}}$. We will assume, as in conventional deployments of wireless systems, that $f_\text{ul}<f_\text{dl}$, thus, $\nu<1$.

We assume that each BS is equipped with a \textit{Uniform Linear Array} (ULA) with $M\gg 1$ antennas with a physical antenna spacing of $d$ and scans the angular range $\Theta=[-\theta_{\max}, \theta_{\max}]$ where $\theta_{\max} \in [0, \frac{\pi}{2}]$. We denote the  response of BS array to a planar wave at frequency $f \in \clF$ by  $\bfa(\theta,f) \in \bC^M$, where 
\begin{align}
[\bfa(\theta,f)]_k = e^{j 2\pi \frac{f}{c_0} kd \sin(\theta)}, k\in [M],
\end{align}
where $c_0$ is the speed of the light and where we defined the short-hand notation {$[M]=\{0,1, \dots, M-1\}$}. For simplicity, we adopt the narrow-band assumption mentioned before, namely, $\frac{W_\text{ul}}{f_\text{ul}}, \frac{W_\text{dl}}{f_\text{dl}} \ll 1$,  and define two array responses $\aul(\theta)=\bfa(\theta, f_{\text{ul}})$ for the UL and $\adl(\theta)=\bfa(\theta, f_{\text{dl}})$ for the DL, where
\begin{align}\label{both_arr_resp}
[\aul(\theta)]_k=e^{j k \frac{2\pi}{\lambda_{\text{ul}}} d\sin(\theta)},\ [\adl(\theta)]_k=e^{j k \frac{2\pi}{\lambda_{\text{dl}}} d\sin(\theta)},
\end{align}
for $ k \in [M]$, where $\lambda_{\text{ul}}=\frac{c_0}{f_{\text{ul}}}$ and $\lambda_{\text{dl}}=\frac{c_0}{f_{\text{dl}}}$ denote the wavelength at the UL and the DL carrier frequencies, and where $\frac{\lambda_{\text{dl}}}{\lambda_{\text{ul}}}=\frac{f_{\text{ul}}}{f_{\text{dl}}}= \nu <1$.  We will assume that the antenna spacing $d$ is set to $d=\varrho \frac{\lambda_{\text{ul}}}{2 \sin(\theta_{\max})}$, where $\varrho \in (0,1)$ is the spatial oversampling factor. Note that since the angular range $\Theta=[-\theta_{\max}, \theta_{\max}]$ scanned by the array  has an angular span of $2 \theta_{\max}$, the antenna spacing $\frac{\lambda_{\text{ul}}}{2 \sin(\theta_{\max})}$ is the minimum one required to avoid spatial aliasing or grating lobes, which is the reason we call $\varrho$ the spatial oversampling factor. In array processing applications, where one deals with only a single frequency band, say, $\clF_{\text{ul}}$, it is conventional to set $\varrho=1$ since this yields the maximum physical span, thus, the maximum angular resolution of the array. In this paper, we  deal with communication at two disjoint frequency bands $\clF=\clF_{\text{ul}} \cup \clF_{\text{dl}}$ and we will assume that $\varrho < \nu$ to avoid  grating lobes \cite{ van2002optimum} in both bands $\clF_{\text{ul}}$ and $\clF_{\text{dl}}$. As we will explain in the sequel, for the problem addressed in this paper,  we will need even smaller values of $\varrho$.

Let us consider a generic user served by the BS. We assume that the communication channel between this user and the BS is through a set of scatterers in the environment. We denote the frequency-dependent  channel vector of the user at time slot $t$ by $\bfh(t,f)=\sum_{i=1}^p w_i(t) \bfa(\theta_i, f)$ where  $\theta_i$ and $w_i(t)$ denote the \textit{angle of arrival} (AoA) and  the random channel gain of the $i$-th scatterers. The channel gains $\{w_i(t): i \in [p]\}$ typically vary quite fast across consecutive time slots but the AoAs $\{\theta_i: i \in [p]\}$ can be safely assumed to remain stable for many time slots.  Invoking the central limit theorem, we will assume that $w_i(t)\sim \cg(0, \gamma_i)$ where $\gamma_i\in \bR_+$ denotes the strength of the $i$-th scatterer. 
We define the covariance matrix of the channel vector at a frequency $f \in \clF$  by
\begin{align}\label{cov_disc}
\Sigmam(f)=\bE[\bfh(t,f) \bfh(t,f)^\herm]=\sum_{i=1}^p \gamma_i \bfa(\theta_i, f) \bfa(\theta_i,f)^\herm.
\end{align}
In this paper, we will adopt a more general model, where the scattering channel might consist of a continuum of scatterers and  the channel vector  $\bfh(t,f)$ is given by
\begin{align}\label{hs_form}
\bfh(t,f)=\int _{-\theta_{\max}}^{\theta_{\max}} W(t,d\theta) \bfa(\theta, f),
\end{align}
where $W(t, d\theta)$ denotes the Gaussian channel gain of those  scatterers with AoAs in $[\theta, \theta+ d \theta)$ at time slot $t$, which is a circularly symmetric stochastic process with independent increments in the angular domain (uncorrelated scattering)\footnote{We refer to \cite{grimmett2001probability} page 387 and \cite{lifshits2013gaussian} page 36 for a more rigorous definition of the random spectral representation (or equivalently white noise representation) of stationary stochastic processes.}:
\begin{align}\label{us_model}
\bE\big [W(t,d\theta) W(t, d \theta')^* \big ]=\gamma(d\theta) \delta(\theta-\theta'),
\end{align}
where $\delta(.)$ denotes the Dirac's delta function and where $\gamma(d\theta)$ is a positive measure denoting the channel strength of those scatterers lying in the angular range $[\theta, \theta+ d\theta)$. We will call $\gamma(d\theta)$ the angular \textit{power spread function} (PSF) of the channel vectors\footnote{The process $W(t, d\theta)$ is fully characterized by its temporal-spatial covariance function $\bE\big [W(t,d\theta) W(t', d \theta')^* \big ]=\beta(t-t') \gamma(d\theta) \delta(\theta-\theta')$ where $\beta(.)$ denotes the temporal correlation function of the process with $\beta(0)=1$. In practice, $\beta (\tau) \approx 0$ for $|\tau| > \Delta t_c$ where $\Delta t_c$ denotes the coherence time of the channel, i.e., the channel vectors are almost independent across different slots separated by larger than a coherence time $\Delta t_c$.}. In general, in a massive MIMO multiuser setting each user is characterized by its own angular PSF, which depends on the propagation geometry of each specific user to the BS array. Since the focus of this paper is UL-DL  covariance interpolation, we focus on a generic PSF $\gamma$, while it is understood that the proposed interpolation scheme and corresponding analysis can be applied to each user in the system.   Using \eqref{us_model}, we can also write the covariance matrix of the channel vector in this general setting as
\begin{align}\label{sig_f_eq}
\Sigmam(f)=\int_{-\theta_{\max}}^{\theta_{\max}} \gamma(d\theta) \bfa(\theta, f) \bfa(\theta, f)^\herm.
\end{align}
In the special case, where $\gamma(d\theta)=\sum_{i=1}^p \gamma_i \delta(\theta-\theta_i)$ is a discrete measure, this reduces to $\Sigmam(f)$ introduced in \eqref{cov_disc}. It is seen that for a given angular PSF $\gamma$, the corresponding channel covariance matrix $\Sigmam(f)$ varies smoothly with $f$. Here, for simplicity, we assume that the statistics of the channel stochastic process $\{\bfh(t,f)\}$ can be well approximated  by $\Sigul=\Sigmam(f_{\text{ul}})$ in the UL  and by $\Sigdl=\Sigmam(f_{\text{dl}})$ in the DL. 

Covariance information, especially in the DL,  is of huge practical use for efficient signal processing in massive MIMO, e.g., for grouping and serving users more efficiently \cite{adhikary2013joint,nam2014joint, haghighatshoar2017massive, haghighatshoar2015channel},  for designing low-complexity beamforming algorithms  \cite{Moshavi1996,Sessler2005, Zarei,Bjornson, benzin2017truncated, boroujerdi2017low}, and for  efficient channel probing and feedback in FDD massive MIMO systems \cite{khalilsarai2017efficient}.  The frequency-variant nature of the  channel covariance matrix is, however, highly overlooked in the current massive MIMO literature. Interestingly, this effect is negligible when the number of antennas $M$ is quite small but plays a crucial role in a massive MIMO regime where $M\gg 1$ such that $\frac{M (f_\text{dl}-f_\text{ul})}{f_\text{dl}} = M(1-\nu)=O(1)$.
\begin{example}\label{examp1}
	Consider a simple line-of-sight scenario where the scattering channel between the user and the BS consists of a line scatterer at AoA $\theta_0 \in \Theta= [-\theta_{\max}, \theta_{\max}]$ and is given by $\bfh_\text{ul}(t)=w(t) \aul(\theta_0)$ in the UL and with  $\bfh_\text{dl}(t)=w(t) \adl(\theta_0)$ in the DL where $w(t)\sim \cg(0, \gamma)$ is the Gaussian channel gain of the scatterer.  Let us consider a scenario, where one neglects the frequency-variant nature of array responses and uses $\aul(\theta_0)$ (as the dominant signal subspace of $\Sigul$) for beamforming in the DL. Due to the UL-DL frequency mismatch, this results in an attenuation factor of order 
	\begin{align*}
	\vartheta(\nu, M)&= \frac{|\aul(\theta_0)^\herm \adl(\theta_0)|}{M}= \frac{\big |\sin\big (M \pi (1-\nu) \frac{\sin(\theta_0)}{\sin(\theta_{\max})} \big ) \big | }{M \big |\sin \big (\pi (1-\nu) \frac{\sin(\theta_0)}{\sin(\theta_{\max})} \big ) \big |},
	\end{align*}
	 compared with the ideal beamforming vector $\adl(\theta_0)$.
	It is seen that $\vartheta(\nu,M) \to 1$ as $\nu=\frac{f_\text{ul}}{f_\text{dl}}\to 1$. However, for any $\theta_0 \neq 0$, the attenuation factor $\vartheta(\nu,M)$ can  potentially approach  $0$ in a massive MIMO scenario when the number of antennas $M$ is very large such that $M(1-\nu)=O(1)$.  \hfill $\lozenge$
\end{example}

\subsection{Contribution}
In this paper, we address the problem of covariance matrix interpolation in frequency in massive MIMO systems (especially, FDD massive MIMO). More specifically, we assume that the covariance matrix $\Sigul$ of a generic user is estimated by the pilot signal transmitted from the users \cite{haghighatshoar2017massive, haghighatshoar2015channel, haghighatshoar2016low, haghighatshoar2017low} in the UL  and provide a procedure to estimate $\Sigdl$ from the available $\Sigul$. Our proposed scheme has the following advantages. First, estimating $\Sigul$ from UL pilots does not impose any extra pilot transmission since the UL pilots are any way needed for serving the users (receiving the data from the users) in the UL. Second, although it is possible to  estimate $\Sigdl$ directly from the samples of the DL channel vector but it requires probing the channel via pilot transmission from the BS to the users in the DL and feedback of the received channel vectors at the user side to the BS in the UL, so that the BS can reliably estimate $\Sigdl$. This incurs a huge feedback overhead, as is well-known in FDD systems \cite{khalilsarai2017efficient}. Our  goal in this paper is to avoid this feedback overhead by estimating/interpolating $\Sigdl$ indirectly through estimating only $\Sigul$. Our proposed technique relies on the following two key assumptions: 
\vspace{1mm}

\noindent
\colorbox{gray!40}{ 
	\begin{minipage}{0.47\textwidth}
		
		\noindent{\bf A1. Stationarity:} We assume that $\gamma(d\theta)$ is locally time-invariant and remains constant over many time slots (channel coherence times), such that the UL covariance matrix $\Sigul$ can be reliably estimated from the UL pilots. Mathematically speaking, we assume that, for each fixed frequency $f$, the channel vector process $\bfh(t,f)$ is a locally (across time slots $t$) stationary process.
		\vspace{2mm}
		
		\noindent{\bf A2. Power Profile Reciprocity:} We assume that the angular PSF $\gamma(d\theta)$ is   frequency-invariant, over the whole frequency band  $\clF=\clF_{\text{ul}} \cup \clF_{\text{dl}}$ consisting of the UL and the DL frequency bands. This can be justified by the fact that the electromagnetic properties of the scatterers in the channel do not change significantly over the  spectrum $\clF$. 				$\phantom{\sum}$  \hfill $\lozenge$
	\end{minipage}
}
\vspace{1mm}

\noindent
It is important to note that the reciprocity assumption {\bf A2}   is quite different than the traditional reciprocity assumption on the instantaneous channel vector in  massive MIMO literature. More specifically, the UL-DL reciprocity in massive MIMO refers to the fact that to what extent the instantaneous channel process in the DL $\bfh_{\text{dl}}(t):=\bfh(t,f_\text{dl})$, seen as a stochastic process, can be estimated/predicted from the observation of the UL channel vector  $\bfh_{\text{ul}}(t):=\bfh(t,f_\text{ul})$ or vice versa. 
Mathematically speaking, this type of reciprocity is indeed possible but requires the quite restricting condition that the power spectral density of the process in the delay-angle domain  be discrete with resolvable angle-delay components, which is barely fulfilled in practical massive MIMO  scenarios. We refer to \cite{8006735, shirani2010mimo,vasisht2016eliminating} for further discussion on the predictability of the stochastic processes and some relevant applications in MIMO systems. The assumption {\bf A2} in this paper, in contrast, deals with the reciprocity of the second order statistics, i.e., the angular PSF or the angular power profile, of the channel vectors  rather than their random realizations (specifically in the frequency domain) and can be safely assumed to hold in almost all practical scenarios. In this paper, we mathematically prove that  under the assumptions {\bf A1} and {\bf A2} and mild conditions on the parameters $\varrho, \nu$, the DL covariance matrix $\Sigdl$ can be stably estimated from the UL one $\Sigmam_\text{ul}$. 
We  propose several off-the-shelf algorithms for UL-DL covariance interpolation and evaluate their performance via numerical simulations. 

\subsection{Notation}
We have already introduced some of the notation. We show vectors by boldface small letters (e.g., $\xv$), matrices by boldface capital letters (e.g., $\Xm$), scalar constants by 
non-boldface letters (e.g., $x$ or $X$), and sets by calligraphic letters (e.g., $\clX$).
We denote the $i$-th row and $j$-th column of a matrix $\bfX$ with a row vector $\Xm_{i,.}$ and a column vector $\Xm_{.,j}$.
We represent the Hermitian and the transpose of a matrix (or a vector) $\bfX$ by $\bfX^\herm$ and $\Xm^\transp$ respectively.
We use $\|\bfx\|$ for the $l_2$-norm of a vector $\bfx$, and $\|\bfX\|$ for the Frobenius norm of a matrix $\bfX$.
We denote the big-O and small-o  by $O(.)$ and $o(.)$ respectively.

\section{Problem Statement}\label{prob_state}
Let us consider a generic user and let us denote the frequency-dependent channel covariance matrix of this user by $\Sigmam(f)$ as in \eqref{sig_f_eq}. Note that for the ULA considered here, $\Sigmam(f)$ is a Hermitian \textit{positive semi-definite} (PSD) Toeplitz matrix whose first column from \eqref{sig_f_eq}  is given by  
\begin{align}
\sigmam(f)=\int_{-\theta_{\max}}^{\theta_{\max}} \gamma(d\theta) \bfa(\theta, f).
\end{align}
We  define $\sigul=\sigmam(f_\text{ul})$ and $\sigdl=\sigmam(f_\text{dl})$ as the first column of $\Sigmam_\text{ul}$ and of  $\Sigmam_\text{dl}$ respectively. We assume, as before, that the physical antenna spacing is $d=\varrho \frac{\lambda_\text{ul}}{2 \sin(\theta_{\max})}$, where $\varrho \in (0,1)$ is the spatial oversampling factor, and make the change of variables $\xi=\frac{ \sin(\theta)}{\sin(\theta_{\max})}$, where we write
\begin{align}\label{sig_dl_int_eq}
\sigul=\int_{-1}^{1} \gamma(d\xi) \aul(\xi), \, \sigdl=\int_{-1}^{1} \gamma(d\xi) \adl(\xi),
\end{align}
where, with some abuse of notation, we denoted  the array responses and the resulting measure in the $\xi$-domain after the change of variable again by  $\aul, \adl$ and $\gamma$ respectively.
Note that the normalized UL/DL array responses are given by $\aul, \adl: [-1,1] \to  \bC^M$, where
\begin{align}\label{aul_dumm_1}
[\aul(\xi)]_k = e^{j k \pi \varrho  \xi}, \, [\adl(\xi)]_k = e^{j k \pi \frac{\varrho}{\nu}  \xi}, \ \ k\in [M].
\end{align}
For simplicity, and without  loss of generality, we  assume that $\gamma(d\xi)$ is a normalized positive measure  with $\gamma([-1,1])=1$.  We  define the continuous Fourier transform of $\gamma$ as $\check{\gamma}: \bR \to \bC$:
\begin{align}\label{f_trans}
\check{\gamma}(x)=\int_{-1}^1 \gamma(d\xi) e^{j  \pi \xi x}, \ x\in \bR.
\end{align} 
 Since $\gamma$ is a normalized measure with a bounded support $[-1,1]$, $\gamf(x)$ is a continuous function of $x$ since 
\begin{align*}
\lim_{h \to 0} |\gamf(x+h)- \gamf(x)| &\leq \lim_{h \to 0} \int_{-1}^1 \gamma(d\xi) |e^{j\pi \xi x}- e^{j\pi \xi (x+h)}|\nonumber\\
&\leq \lim_{h \to 0} \int_{-1}^1 \gamma(d\xi) |1- e^{j\pi \xi h}|=0,
\end{align*}
where the last expression follows from the dominated convergence theorem as $|1-e^{j \pi \xi h}| \leq 2$ is a bounded function approaching to $0$ for all $\xi \in [-1,1]$ as $h \to 0$. Similarly, we can check that $\gamf(0)=1$ and $|\gamf(x)| \leq |\gamf(0)|$ for all $x \in \bR_+$, and that $\gamf(x)$ has conjugate symmetry, namely, $\gamf(-x)^*=\gamf(x)$ for any $x\in \bR$. Moreover, being the Fourier transform of a positive measure, it is also a \textit{positive definite} function, i.e., $\sum_{i,j=1}^l c_i c_j^* \gamf(x_i-x_j)\geq 0$ for any $l$, any $\{x_i\}_{i=1}^l \subset \bR$, and any sequence of complex numbers $\{c_i\}_{i=1}^l$; We refer to \cite{van2012harmonic} for a comprehensive introduction to positive definite functions and their applications, and to \cite{berlinet2011reproducing} for the connection with \textit{reproducing kernel Hilbert spaces} (RKHS). Also, seen as a function of $x \in \bR$, $\gamf(x)$ is a band-limited function with a bounded spectrum  in $[-1, 1]$. 

Now let us consider $\sigul$. From \eqref{sig_dl_int_eq} and \eqref{aul_dumm_1}, it is seen that 
\begin{align}\label{f_tr_dumm}
[\sigul]_k= \int_{-1}^1 \gamma(d\xi) e^{j k \pi \varrho  \xi}=\check{\gamma}(k \varrho),\ \  k \in [M].
\end{align}
Hence,  the samples $\{\gamf(k \varrho): k \in [M]\}$ of $\check{\gamma}$ at the lattice sampling points $\{k \varrho: k\in [M]\}$ correspond to the $M$ elements of   $\sigul$.
Similarly, it is not difficult to check  that  $[\sigdl]_k=\gamf(\frac{k\varrho}{\nu})$ where $\nu=\frac{f_\text{ul}}{f_\text{dl}}<1$ is the ratio between the UL and the DL carrier frequencies as defined before. This implies that $\sigdl$, and as a results $\Sigdl$, can be obtained from the samples of $\gamf$ at  positions $\{ \frac{k\varrho}{\nu}: k \in [M]\}$. With this explanation, we can pose the problem of UL-DL covariance interpolation as follows. 
\begin{problem}\label{main_prob}
Given the set of $M$ UL samples $\{\check{\gamma}(k\varrho): k \in [M]\}$ of the band-limited function $\check{\gamma}$, with an unknown positive spectrum $\gamma(d\xi)$ supported over $ [-1, 1]$, find the values of the corresponding DL samples 
$\{\check{\gamma}(\frac{k\varrho}{\nu}): k\in [M]\}$.  \hfill $\lozenge$
\end{problem}

This is illustrated in Fig.\,\ref{fig:graph}.  Note that since by our assumption, $\gamf(x)$ is band-limited, from Shannon-Nyqvist  sampling theorem\footnote{As a brief note, we would like to mention that here, for convenience, we defined the Fourier transform  in \eqref{f_trans} by $\pi \xi x$ rather than the conventional $2\pi \xi x$, thus, the  bandwidth of $\gamf$ in the conventional notation is $\frac{1}{2}$ (rather than $1$). Thus, samples of $\gamf$ at $\bZ_+$ have a sampling rate equal to (more than when $\varrho<1$) twice the bandwidth of $\gamma$ and are sufficient for its recovery according to the sampling theorem.} \cite{unser2000sampling}, we  should be able to recover $\gamma$ from the samples $\{\gamf(k \varrho): k \in \bZ_+\}$ (even without any spatial oversampling, i.e., for  $\varrho=1$), thus, to estimate $\gamf(x)$ at any arbitrary $x\in \bR_+$. 
When we have only finitely many samples $\{\gamf(k \varrho): k \in [M]\}$, given the band-limitedness and smoothness of $\gamf$, we may still expect to estimate $\gamf(x)$ for those $x$ inside the UL sampling interval $[0, M\varrho]$ with a moderately small error that vanishes as $M \to \infty$. However, in UL-DL covariance interpolation, there is always a subset of DL sampling interval $[0, \frac{M\varrho}{\nu}]$ that lies near the boundary of  UL sampling interval $[0, M \varrho]$ (see UL/DL sampling intervals in Fig.\,\ref{fig:graph}).
In fact, one can argue that no matter how large $M$ is, as long as $\gamma$ is not completely trivial, those boundary points  suffer from some interpolation error and cannot be approximated very well from the UL samples $\{\gamf(k \varrho): k \in [M]\}$. This implies that for a suitable UL-DL covariance interpolation, we need to apply some truncation at those DL sampling points that lie on the  boundary. This, of course, creates an error in the resulting estimate of the DL covariance matrix. However, when $\check{\gamma}(x)$ decays quite fast in terms of $x$ and $M$ is sufficiently large, we expect that the samples of $\check{\gamma}$ close to the boundary of $[0, M\varrho]$ have a very small amplitude (energy) and contribute negligibly to the  DL covariance matrix. Therefore, overall, we expect that a suitable interpolation of DL samples followed by an appropriate truncation yield a stable UL-DL covariance interpolation for a sufficiently large $M$.

It is also important to note that our explanation in this section  confirms that for a large umber of antennas $M$, the UL and DL covariance matrices can differ significantly from each other since they are obtained by sampling $\gamf(x)$ at quite different sampling intervals (see, e.g., Fig.\,\ref{fig:graph}), so an appropriate covariance interpolation from UL to DL is inevitable in relevant applications (see also Example \ref{examp1}).

\begin{figure}[t]
	\centering
	\includegraphics[scale=1]{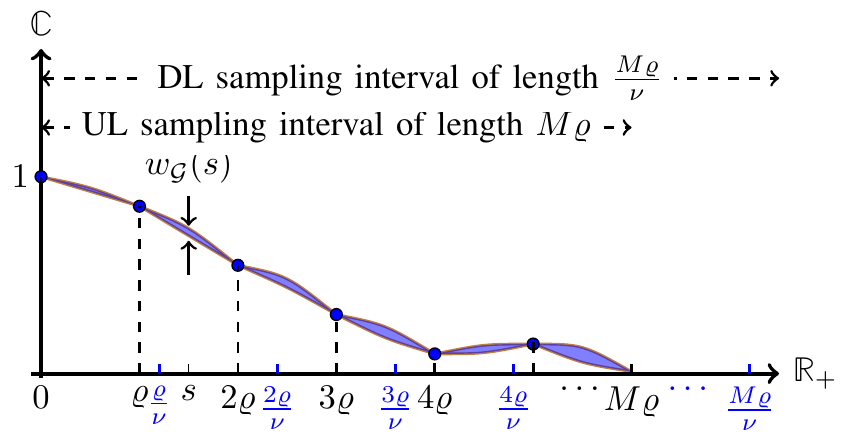}
	\caption{Illustration of the graph $\clG \subset [0, M \varrho] \times \bC$ associated to a specific measure and its width at a specific probing point $s$. Note that the vertical axis corresponds to the complex plane or more precisely to the interior of the complex unit sphere in the complex plane as $|\gamf(x)| \leq |\gamf(0)|=1$, and due to the conjugate symmetry we have plotted $\gamf(x)$ only for $x \in \bR_+$.}
	\label{fig:graph}
\end{figure}

\section{Summary of the Results}\label{sec:summary}
In this section, we briefly  overview our results and discuss some ot its implications in massive MIMO systems. 

\subsection{Basic Setup}
We assume that the angular PSF of the scattering channel is given by a normalized measure $\gamma$ supported in $[-1,1]$. As stated in Section \ref{prob_state}, the UL-DL covariance interpolation can be posed as the problem of estimating the samples $\{\gamf(\frac{k \varrho}{\nu}): k \in [M]\}$ from the observation of the samples $\{\gamf(k \varrho): k \in [M]\}$. Since the underlying measure $\gamma$ is unknown, in this paper, we will focus on a minimax approach to the interpolation problem.

We first define $\Gamma_\gamma$ as the set of all positive normalized measures $\mu \in \Gamma_\gamma$ that are supported on $[-1, 1]$ and yield the same UL covariance matrix as $\gamma$, or more specifically:
\begin{align}\label{Gamma_set}
\Gamma_\gamma:=\Big \{\mu: \mu([-1, 1])=1,  \check{\mu}(k \varrho)=\gamf(k \varrho), k\in[M]\Big \},
\end{align} 
where $\check{\mu}$ denotes the Fourier transform of $\mu$ as in \eqref{f_trans}. We consider the UL probing window $[0, M \varrho] \subset \bR_+$ (see, e.g. Fig.\,\ref{fig:graph}) and define the image of the set $\Gamma_\gamma$ over the probing window $[0, M \varrho]$ under the Fourier transform as 
\begin{align}\label{graph}
\clG:=\bigcup_{\mu \in \Gamma_\gamma} \Big \{(x,\check{\mu}(x)): x \in [0, M \varrho] \Big \} \subset [0, M \varrho] \times \bC,
\end{align}
which is given as the union of the graph of the Fourier transforms of all $\mu \in \Gamma_\gamma$. We will typically consider the probing set $[0,M \varrho]$ in the sequel.

We define the section of the graph $\clG$ at a probing point $s\in [0, M \varrho]$ by $\clG_s=\{z\in \bC: (s,z) \in \clG\}$. We also define the width of $\clG$ at a point $s \in [0, M \varrho]$ as the diameter of $\clG_s$ defined by 
\begin{align}\label{graph_width}
w_\clG(s):= \sup_{a,b \in \clG_s} |a-b|.
\end{align}
Fig.\,\ref{fig:graph} illustrates the graph associated with a specific measure $\gamma$ and its width at a specific point $s\in [0, M \varrho]$. It is important to note that since all the measures in $\Gamma_\gamma$ have the same Fourier transform  at the sampling points $\{k\varrho: k \in [M]\}$, we have that $w_\clG(s)=0$ for $s \in \{k \varrho: k\in [M]\}$. Also, note that $w_\clG(s)$ intuitively measures the variation in the Fourier transform of the measures in $\Gamma_\gamma$ at a specific point $s$. As a result, by measuring $w_\clG(s)$, we can obtain an estimate of the worst-case  error of an  algorithm that estimates $\gamf(s)$ by merely observing $\gamf$ at the sampling points $\{k \varrho: k \in [M]\}$. It is worthwhile to mention that in our setting, intuitively speaking, controlling the width of the graph $w_\clG$ over the probing window conditioned on the available UL samples $\{\gamf(k \varrho): k \in [M]\}$  resembles establishing a \textit{Restricted Isometry Property} (RIP) in the Compressed Sensing setup \cite{donoho2006compressed, candes2006near} by taking sufficiently many linear projections. In other words, the width of the graph specifies to what extent the Fourier transform of two different measures can differ from each other provided that they take on the same values over the UL sampling points $\{k \varrho: k \in [M]\}$.  We will use this as the key ingredient (similar to the RIP in Compressed Sensing) to characterize the robustness/stability of the UL-DL covariance interpolation addressed in this paper. 

\subsection{Main Result and Algorithmic Implications}\label{sec:main}

\noindent With this brief explanation, we can  state our main result.

\begin{theorem}\label{main_thm}
	Let $\gamma$ be an arbitrary positive and normalized measure supported in $[-1, 1]$. Let $[0,M \varrho]$ be the UL probing window  and let $s\in [0, M \varrho]$ be an arbitrary point. Let $\clG$ be the graph corresponding to $\gamma$ and let $w_\clG(s)$ be the width of $\clG$ at $s \in [0, M \varrho]$. Then, there is a universal constant $C$ such that 
	\begin{align}
	w_\clG(s) \leq \min \left \{ C \Big(\sin(\frac{\pi \varrho}{2}) g(\frac{s}{M \varrho}) \Big )^{2M}, 2\right \},
	\end{align}
	 where $g: [0,1] \to [1,2]$ is a universal function independent of $\gamma$, $M$, and $s$, and given explicitly by $g=e^{f}$ where 
	\begin{align}
	f(\alpha)= \Big(1- h_2(\frac{1+\alpha}{2})\Big) \log(2),
	\end{align}
	for $\alpha \in [0,1]$, where $h_2(x)=-x\log_2(x)-(1-x)\log_2(1-x)$ is the binary entropy function for $x \in [0,1]$. \hfill $\square$
\end{theorem}
\begin{figure}[t]
	\centering
	\includegraphics[scale=0.8]{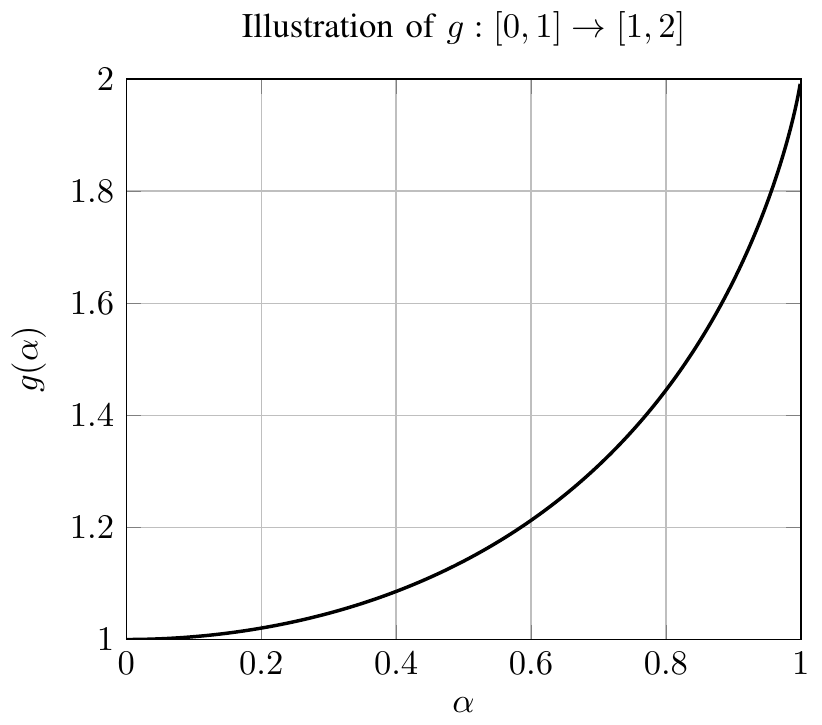}
\caption{Illustration of the function $g$ in $\alpha \in [0,1]$.}
\label{fig:g}
\end{figure}
Fig.\,\ref{fig:g} illustrates the function $g$  in the statement of Theorem \ref{main_thm} in the interval $\alpha \in [0,1]$. 
Let us discuss some of the algorithmic implications of Theorem \ref{main_thm} for the UL-DL covariance estimation that we address in this paper. Consider any arbitrary algorithm that produces an estimate $\mu$ of the true PSF from the observations $\{\gamf(k\varrho): k \in [M]\}$ that satisfies $\{\check{\mu}(k\varrho)=\gamf(k\varrho): k \in [M]\}$. Let $s\in [0, M \varrho]$ be an arbitrary probing point, and let $\gamf(s)$ and $\check{\mu}(s)$ be the Fourier transform of $\gamma$ and the that of the estimate $\mu$ at $s$. Theorem \ref{main_thm}  implies that 
	\begin{align*}
	|\gamf(s)- \check{\mu}(s)| \leq w_\clG(s) \leq \min \left \{ C \Big(\sin(\frac{\pi \varrho}{2}) g(\frac{s}{M\varrho}) \Big )^{2M}, 2\right \}
	\end{align*}
decays exponentially fast to $0$ as $M$ tends to infinity, provided that  $\sin(\frac{\pi \varrho}{2}) g(\frac{s}{M\varrho}) <1$. Moreover, since $g$ is an increasing function in $[0,1]$, we also have that for any  $\clW_0 \subseteq [0, M \varrho]$
\begin{align*}
\sup_{s \in \clW_0} |\gamf(s)- \check{\mu}(s)|& \leq \min \left \{ C \sup_{s\in \clW_0}\Big (\sin(\frac{\pi \varrho}{2}) g(\frac{s}{M \varrho}) \Big ) ^{2M}, 2\right \}\nonumber\\
&= \min \left \{ C \Big (\sin(\frac{\pi \varrho}{2}) g(\frac{s_{\max}}{M \varrho}) \Big ) ^{2M}, 2\right \},
\end{align*}
where $s_{\max}=\sup \{s: s\in \clW_0\}$. Thus, the worst case error over any probing window $\clW_0$ can be controlled by the largest element $\sup \{s: s\in \clW_0\}$ of $\clW_0$. In particular, since $g(\alpha) \in [1,2]$, the estimation is precise over the whole window $[0,M\varrho]$ when $\sin(\frac{\pi \varrho}{2}) \leq \frac{1}{2}$ or equivalently when $\varrho \leq \frac{1}{3}$. For example, in a practical case, where the antenna scans the angular range $\Theta=[-\theta_{\max}, \theta_{\max}]$ with a $\theta_{\max}=60$ degrees, this would require an antenna spacing of $d=\frac{\lambda_{\text{ul}}}{3 \sqrt{3}}$ where $\lambda_{\text{ul}}$ denotes the wavelength at the UL carrier frequency $f_\text{ul}$. 

Overall, since the elements of $\sigdl$ correspond to the samples of $\gamf$ at locations $s\in \{\frac{k\varrho}{\nu}: k \in [M]\}$,  from the condition $\sin(\frac{\pi \varrho}{2}) g(\frac{s}{M\varrho}) <1$ needed for $w_\clG(s) \to 0$ asymptotically as $M \to \infty$, it immediately results that for any $\varrho<1$, one can stably estimate from $\sigul$ those components of $\sigdl$ with indices  belonging to 
\begin{align}\label{feas_index}
\clI_\text{dl} (\varrho):=\{k \in [M]: k \leq M \nu, \, \sin(\frac{\pi \varrho}{2}) g(\frac{k}{M \nu}) <1\}.
\end{align}
Since $|\clI_\text{dl} (\varrho)| \leq M \nu$, the result of Theorem \ref{main_thm}  guarantee the stable recovery of only a fraction of components of $\sigdl$ consisting of the first $|\clI_\text{dl} (\varrho)|$ elements (see Fig.\,\ref{fig:graph} and Fig.\,\ref{fig:graph_tradeoff}). Our interpolation algorithm is summarized in Algorithm \ref{tab:UL_DL_interp}.

\begin{algorithm}[h]
	\caption{ UL-DL Covariance Interpolation} 
	\label{tab:UL_DL_interp} 
	\begin{algorithmic}[1]
		\State {\bf Input:} $M$, $\varrho$, $\nu$, an estimate of the first column of the UL covariance matrix $\Sigul$ given by $\sigul$ with $[\sigul]_0=1$.
		\State Find the index set $\clI_\text{dl}(\varrho)$ as in \eqref{feas_index}.
		\State Find an arbitrary positive normalized measure $\mu$ with $\check{\mu}(k \varrho)=[\sigul]_k$ for $k \in [M]$.
		\State Find from $\mu$  an estimate of the first column of the DL covariance matrix $\Sigdl$ as $[\sigdl]_k=\check{\mu}(\frac{k\varrho}{\nu})$ for $k \in [M]$.
		\State Keep the elements $[\sigdl]_k$ for $k \in \clI_\text{dl} (\varrho)$ as they are and set to zero the elements $[\sigdl]_k$ for $k \not \in \clI_\text{dl} (\varrho)$. 
		\State {\bf Output:}  The Toeplitz matrix corresponding to the  truncated estimate $\sigdl$.
	\end{algorithmic}
\end{algorithm}

\subsection{Underlying Trade-offs}\label{sec:trade}
\begin{figure}[t]
	\centering
	\includegraphics[scale=0.9]{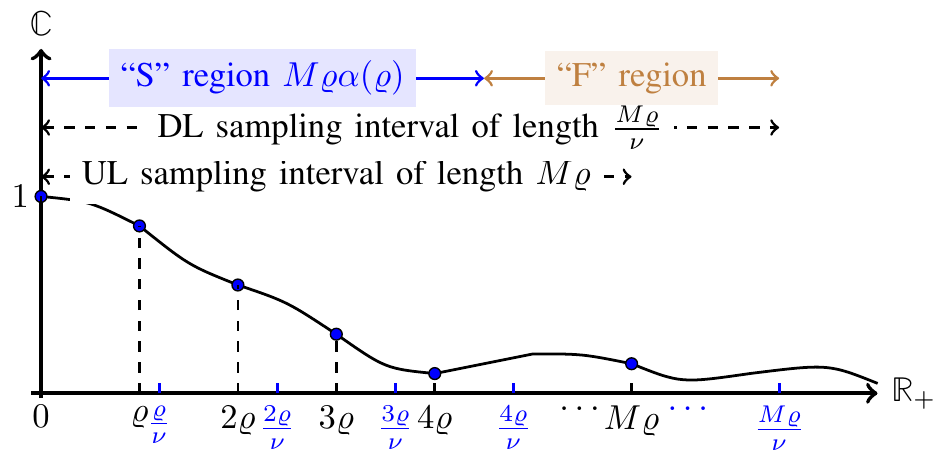}
\caption{Illustration of the variation in sampling by changing the spatial oversampling factor $\varrho \in [0,1]$. 
It is seen that for a fixed $M$ the observation window in the UL and the DL has a length proportional to $\varrho$. Also, for any $\varrho$, only a fraction of DL coefficients in  a region of size $M \varrho \alpha(\varrho)$ denoted by ``S'' region can be robustly estimated from the UL samples, where a robust estimation of the rest of the DL coefficients that lie in a window denoted by ``F'' region can generally fail. }
\label{fig:graph_tradeoff}
\end{figure}
The result of Theorem \ref{main_thm} might be misleading since it seems to suggest that one needs to select a smaller $\varrho$ to obtain a better interpolation performance. However, one should note that selecting a small $\varrho$ creates a significant  spatial correlation among the antennas and reduces the spatial \textit{degrees-of-freedom} (DoF) of the array. In words, for a given $\varrho$, the  spatial resolution, thus, the number of effective spatial eigen-functions of the array scales as $O(M \varrho)$. For example, in the extreme case of $\varrho \to 0$, all the antennas are collocated and it is as if having only a single antenna. We refer to \cite{poon2005degrees} for a more rigorous explanation of the spatial DoF of the array and related information-theoretic trade-offs.
Here, we only provide an intuitive explanation and mainly focus on the non-asymptotic regime with a finite number of antennas $M$. Consider a specific angular PSF $\gamma$ supported on $[-1, 1]$.  
Using the error bound in Theorem \ref{main_thm}, we see that we can robustly estimate, from UL observations $\{\check{\gamma}(k \varrho): k \in [M]\}$ obtained via UL covariance estimation, only those DL coefficients $\{\check{\gamma}(k \frac{\varrho}{\nu}): k \in [M]\}$ that lie inside the window $[0, M \varrho \alpha]$ provided that  $\alpha\in [0,1]$ satisfies $\sin(\frac{\pi \varrho}{2}) g(\alpha) <1$. We define the largest such $\alpha$ by $\alpha(\varrho)=\sup\{\alpha: \sin(\frac{\pi \varrho}{2}) g(\alpha)<1\}=g^{-1}(\frac{1}{\sin(\frac{\pi \varrho}{2})})$. Thus, for any $\varrho$ we can only guarantee the robust estimation of those DL coefficients inside the window $[0, M \varrho \alpha(\varrho)]$ consisting of the first 
\begin{align}\label{N_coef}
N(\varrho) = \frac{M \varrho \alpha(\varrho) }{\frac{\varrho}{\nu}}= M \nu g^{-1}(\frac{1}{\sin(\frac{\pi \varrho}{2})}),
\end{align}
sampling points of the DL coefficients $\{\check{\gamma}(k \frac{\varrho}{\nu}): k \in [M]\}$. This has been illustrated in Fig.\,\ref{fig:graph_tradeoff}. Intuitively speaking, for any $\varrho$, we have $M \varrho$ spatial DoF in the UL among which only a fraction of $\alpha (\varrho) \in [0,1]$ can be robustly estimated and used for the DL, thus, a total of $M D(\varrho)$ robust DoF for the DL where $D(\varrho)=\varrho \alpha(\varrho)$. Fig.\,\ref{fig:dof} illustrates $D(\varrho)$ for $\varrho \in [0.5,1]$, where its is seen that the maximum $D(\varrho)$ for $\varrho \in (0,1)$ is achieved at $\varrho \approx 0.5$. For a ULA with $\theta_{\max}=60$ degrees, this corresponds to an antenna spacing of $d=\frac{\sqrt{3}}{3} \frac{\lambda_{\text{ul}}}{2}$.

\begin{figure}[t]
\centering
\includegraphics[scale=0.8]{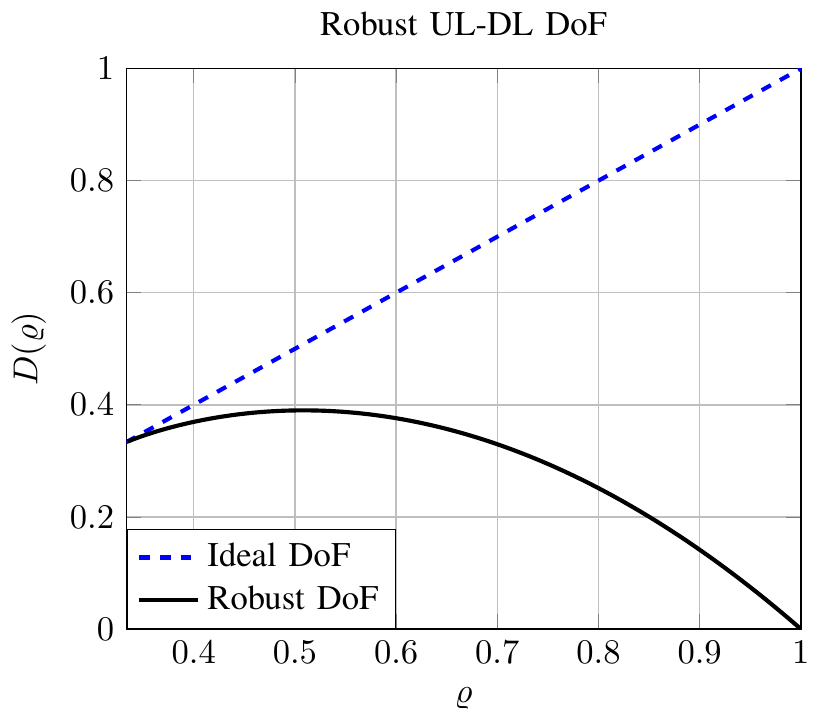}
\caption{Comparison between the ideal and achievable (robust) UL-DL DoF (due to frequency variation) as a function of spatial oversampling factor $\varrho$. It is seen that for $\varrho \leq \frac{1}{3}$, both DoFs are equal. Moreover, the maximum robust DoF is achieved for $\varrho\approx 0.5$. }
\label{fig:dof}
\end{figure}

We should also point out that the result stated in Theorem \ref{main_thm} and also our explanation in this section follow from the minimax analysis of the interpolation problem where we consider all possible angular PSFs and even the worst-case ones. The situation is much better in practice when one deals with a much more structured class of angular PSFs $\gamma$. For example, consider the class of all PSFs $\gamma$ whose Fourier transform $\gamf(x)$  has a negligible energy beyond $|x| > \Delta $ for some fixed $\Delta>0$. In such a case, we expect that  all the significant components of $\gamf(x)$ in $|x|\leq \Delta$ be recovered precisely if $M \varrho \alpha(\varrho) \approx \Delta$ (see, e.g., Fig.\,\ref{fig:graph_tradeoff}), where one can essentially ignore (i.e., zero pad) the coefficients in $|x|> \Delta$ since they are negligibly small. Note that the condition $M \varrho \alpha(\varrho) \approx \Delta$  can be fulfilled even with a moderately large number of antennas. Moreover, for any fixed $\Delta$, one can select $\varrho$ very close to $1$ (no spatial oversampling, thus, maximum spatial resolution) as $M \to \infty$, thus,  reconstructing the DL covariance matrix  without paying any penalty in spatial DoF.

\subsection{Changing the UL and DL frequency bands}
In this paper,  we assumed that, as in almost all massive MIMO deployments, $f_\text{ul} < f_\text{dl}$, thus, $\nu=\frac{f_\text{ul}}{f_\text{dl}}<1$. In such a case, the desired interpolation window $[0, \frac{M \varrho}{\nu}]$ needed for UL goes much beyond the observation window $[0, M \varrho]$ in the DL (see, e.g., Fig.\,\ref{fig:graph} and Fig.\,\ref{fig:graph_tradeoff}), which makes a reliable UL-DL covariance interpolation quite challenging.  In contrast, when $f_\text{ul} > f_\text{dl}$ and $\nu=\frac{f_\text{ul}}{f_\text{dl}}>1$, the interpolation window $[0, \frac{M \varrho}{\nu}]$ for the DL is a smaller subset of observation window $[0, M \varrho]$ in the UL. As a result, the UL-DL covariance interpolation is indeed much easier, even with a moderate number of antennas $M$, and can be done with much less spatial oversampling, i.e., with a larger $\varrho$. 

\section{Chebyshev Differential Equation}\label{sec:cheby}
In this section, we develop suitable approximation techniques and derive  upper bounds on the approximation error using Chebyshev's functions. We will use these results in Section \ref{sec:minimax_analysis} to obtain suitable  minimax upper bounds and to finally prove Theorem \ref{main_thm}. 
\subsection{Chebyshev Equation}
We first consider the Chebyshev second order \textit{ordinary differential equation} (ODE) given by 
\begin{align}\label{cheb_ode}
(1-t^2)y'' - t y' + \varsigma^2 y=0,
\end{align}
where $\varsigma \in \bR_+$ is a fixed parameter. Since the coefficients of the  ODE \eqref{cheb_ode} are differentiable infinitely many times in a  neighborhood of $t=0$, the ODE has an analytic solution as a power series $y(t)=\sum_{n=0}^\infty a_n t^n$ around $t=0$. Moreover, since the ODE \eqref{cheb_ode} has singular points at $t=\pm 1$, from standard results in ODE \cite{boyce1969elementary, coddington1955theory},  it results that this series solution has a convergence radius of at most $1$ around $t=0$. As we will see in Section \ref{sec:minimax_analysis}, we will need the solutions of this ODE in the interval $[-\eta, \eta]$ for $\eta=\sin(\frac{\pi \varrho}{2})$, which will be included in the region  $(-1,1)$ as $\varrho\in (0,1)$.
Replacing the power series in \eqref{cheb_ode}, we obtain the following recursion for the coefficients $a_n$: 
\begin{align}\label{rec_eq}
a_{n+2}=\frac{n^2 -\varsigma^2}{(n+1)(n+2)} a_n.
\end{align}
It is seen that the resulting recursion has order $2$, thus, it yields two linearly independent solution $y_e(t)$ and $y_o(t)$ for the initial values $(a_0=1, a_1=0)$ and $(a_0=0, a_1=1)$. In fact, $y_e$ and $y_o$ are even and odd functions of $t$ respectively, thus, they are linearly independent and span the two-dim space (as the ODE is second order) of the solutions of the ODE \eqref{cheb_ode}, namely, any arbitrary solution $y(t)$ can be written as a linear combination of $y_e$ and $y_o$. It is also worthwhile to mention that when $\varsigma=2k_0$ is an even integer, $y_e(t)$ is an even polynomial of degree $2k_0$, whereas $y_o(t)$ has infinitely many terms in its power series. Similarly, when $\varsigma=2k_0+1$ is an odd integer, $y_o(t)$ is an odd polynomial of order $2k_0+1$, whereas $y_e(t)$ has infinitely many terms in its power series. These polynomial solutions correspond to Chebyshev polynomials of even and odd order. 

\subsection{Explicit Formula for the Solutions and Error Bounds}
A direct calculation shows that $\cos( \varsigma\sin^{-1}(t))$ and $\sin( \varsigma\sin^{-1}(t))$ satisfy the Chebyshev ODE with a parameter $\varsigma$. Since $\cos( \varsigma\sin^{-1}(t))$ is an even function over $t \in (-1,1)$,  has power series expansion around $t=0$,  and satisfies $\cos( \varsigma\sin^{-1}(t))\left |_{t=0}=1\right.$ and $\frac{d}{dt} \cos( \varsigma\sin^{-1}(t))\left |_{t=0}=0\right.$, from $y_e(0)=1, y_e'(0)=0$ and the uniqueness of the solutions of ODE \eqref{cheb_ode}, it should correspond to $y_e(t)$. Similarly,  we can check that $y_o(t)=\frac{\sin( \varsigma\sin^{-1}(t))}{\varsigma}$. We will mainly focus on $y_e(t)$ with a parameter $\varsigma=2 s$ for $s\in [0,M]$. We have 
\begin{align}
y_e(t)=1+ \sum_{k=1}^\infty a_{2k} t^{2k},
\end{align} 
where $a_{2k}$ from \eqref{rec_eq} is given by  
\begin{align}\label{a_recur}
a_{2k}(s)= \prod _{n=0}^{k-1} \frac{(2n)^2 - (2s)^2}{(2n+1)(2n+2)},
\end{align}
where we also represented explicitly the dependence of $a_{2k}$ on $s$. We will keep $s$ fixed in this section and will drop the explicit dependence on $s$ for notation simplicity. 
We first write $y_e(t)$ as follows
\begin{align*}
y_e(t)=1+ \sum_{k=1}^{M-1} a_{2k} t^{2k} + \sum_{M}^\infty a_{2k} t^{2k}=: y_{e,M}(t) + E_M(t),
\end{align*}
where $E_M(t)$ denotes the truncation error consisting of the terms with exponents larger than or equal to $2M$. We first prove the following key result.
\begin{proposition}\label{error_bound_prop}
	Let $E_M(t)$ be the truncation error of order $M$ as define before. Then, we have the following:

\noindent {\bf 1.} $E_M(t)=(-1)^{\floor{s}} \sum_{k=M}^\infty |a_{2k}| t^{2k}$ for all $t \in (-1,1)$, where $\floor{s}$ denotes the largest integer smaller than $s$. In particular, $E_M(t)$ has the same sign for all $t \in (-1,1)$.
	
\noindent 	{\bf 2.} For any fixed $\eta \in (0,1)$, the  error $E_M(t)$ converges to $0$ uniformly for $t \in [-\eta, \eta]$ with the maximum error occurring at the boundary $t=\pm \eta$, i.e., $\max_{t \in [-\eta, \eta]} |E_M(t)|=|E_M(\pm \eta)|$.
	
\noindent 	{\bf 3.} Over the interval $t \in [-\eta, \eta]$, the truncation error $E_M(t)$ is upper bounded  by $a_{2M} \frac{\eta^{2M}}{1-\eta^2}$. 

\noindent 	{\bf 4.} For any fixed $\eta \in (0,1)$, the derivative of the truncation error $E_M(t)$ also converges uniformly to $0$ in the interval $t \in [-\eta, \eta]$ and satisfies 
$\max_{t \in [-\eta, \eta]} |E_M'(t)| \leq 2|a_{2M}| \frac{\eta^{2M-1} (M- (M-2) \eta)}{(1-\eta^2)^2}$. \hfill $\square$
\end{proposition}

\begin{proof}
	To prove part 1, first note from \eqref{a_recur} that the coefficients $a_{2k}(s)$ have alternating signs for $k \leq s$ (due to multiplication by the negative factor $(2k)^2-(2s)^2$), whereas all the coefficients with $k>s$ have the same sign (since $(2k)^2-(2s)^2$ is positive). As $s\in [0,M]$, this implies that all the coefficients $a_{2k}$ for $k\geq M$ have the same sign which can be checked to be $(-1)^{\floor{s}}$. As a result, we can write $E_M(t)=(-1)^{\floor{s}} \sum_{M}^\infty |a_{2k}| t^{2k}$. 
	
	To prove part 2, note that from part 1 and the fact that $t^{2k}$ are increasing functions of $t^2$, it immediately results that the maximum of $E_M(t)$ over $t \in [-\eta, \eta]$ is achieved at the boundary point $ t=\pm \eta$. Since $\eta \in (-1,1)$,  from the convergence of the series at $\eta$, it results that $|E_M(\pm \eta)|$ converges to $0$ as $M$ tends to infinity. This implies the uniform convergence of $E_M(t)$ to $0$ for all $t \in [-\eta, \eta]$ as $M\to \infty$. 
	
	To prove part 3, note that from the recursion equation for the coefficients $a_{2k}$ in \eqref{a_recur}, we have $\frac{a_{2k+2}}{a_{2k}}=\frac{(2k)^2 -(2s)^2}{(2k+1)(2k+2)}$. Since $s\in [0,M]$, it is seen that $|\frac{a_{2k+2}}{a_{2k}}| \leq 1$ for $k\geq M$, thus, $|a_{2k}|$ is a decreasing sequence of $k$ for $k \geq M$. As a result, over the interval $t \in [-\eta, \eta]$, we have that
	\begin{align}
	\max_{t \in [-\eta, \eta]} |E_M(t)|&=|E_M(\pm \eta)|=\sum_{k=M}^\infty |a_{2k}| \eta^{2k} \\
	&\leq |a_{2M}| \sum_{k=M}^\infty \eta^{2k}= |a_{2M}| \frac{\eta^{2M}}{1-\eta^2}.
	\end{align}
	Finally to prove the last part, note that we have $E_M'(t)=(-1)^{\floor{s}} \sum_{k=M}^\infty |a_{2k}| (2k)t^{2k-1}$, thus, $|E_M'(t)|=\sum_{k=M}^\infty |a_{2k}| (2k)t^{2k-1}$. This implies that 
	\begin{align}
	|E_M'(t)| \leq\sum_{k=M}^\infty |a_{2k}| (2k)\eta^{2k-1}=|E_M'(\eta)|,
	\end{align}
	for $t \in [-\eta, \eta]$, thus, the maximum of $|E_M'(t)|$ is achieved at the boundary point $t=\pm \eta$. Moreover, by applying the ratio test \cite{rudin1964principles}, we can see that multiplication of the coefficients $a_{2k}$ by $2k$ does not change the radius of the convergence of the series, thus, $|E_M'(\pm \eta)|$ converges to zero as $M\to \infty$, which implies the uniform convergence of $E'_M(t)$ in the interval $[-\eta, \eta]$. Also, to obtain the final expression, note that 
	\begin{align}
	|E_M'(t)|&\leq |E_M'(\pm \eta)|=\sum_{k=M}^\infty |a_{2k}| (2k) \eta^{2k} \\
	&\stackrel{(a)}{\leq} |a_{2M}| \sum_{k=M}^\infty (2k) \eta^{2k-1}\\
	&=|a_{2M}| \frac{d}{d\eta} \sum_{k=M}^\infty \eta^{2k}\\
	&=|a_{2M}| \frac{d}{d\eta} \frac{\eta^{2M}}{1-\eta^2}\\
	&=2|a_{2M}| \frac{\eta^{2M-1} (M- (M-2) \eta)}{(1-\eta^2)^2}
	\end{align}
	where in $(a)$ we used the fact that $|a_{2k}|$ is a decreasing sequence of $k$ for $k\geq M$.
	This completes the proof.
\end{proof}

\subsection{More Refined Error Analysis}
In this part, we will focus on the  scaling of the coefficient $|a_{2M}(s)|$ as a function of $s$. Our goal is to find a scaling law of an exponential form  $|a_{2M}(s)|\leq g(\frac{s}{M})^{2M}$ for some continuous function $g: [0,1] \to \bR_+$ as given in Fig.\,\ref{fig:g}. We will combine this with the the error bound derived in Proposition \ref{error_bound_prop} to prove that the truncation error will vanish for all $s\in[0,M]$ inside the probing window for which $g(\frac{s}{M}) \eta <1$. We will make this more  rigorous in the following. 
From \eqref{a_recur}, we have  
\begin{align}
|a_{2M}(s)|&= \prod _{n=0}^{M-1} \frac{|(2n)^2 - (2s)^2|}{(2n+1)(2n+2)}\\
&=\frac{(2M)^{2M}}{(2M)!} \prod _{n=0}^{M-1} |(\frac{n}{M})^2 - (\frac{s}{M})^2|.
\end{align}
Applying the Stirling's approximation for a positive integer $l$
\begin{equation}
\sqrt{2\pi l} (\frac{l}{e})^l \leq l! \leq e\sqrt{l} (\frac{l}{e})^l,
\end{equation}
 we can upper/lower bound $h(s):=\frac{1}{2M} \log|a_{M}(s)|$ as follows
\begin{align}
 f(s) - \frac{\log(2e^2 M) }{2M} \leq h(s) \leq  f(s)- \frac{\log(4 \pi M)}{2M}  ,
\end{align}
where we defined $f(s)$ as the following function 
\begin{align}
f(s)=1+ \frac{1}{2M} \sum_{n=0}^{M-1} \log|(\frac{n}{M})^2 - (\frac{s}{M})^2|.
\end{align}
We focus on a scaling regime where $s$ and $M$ grow proportionally such that  $\frac{s}{M} \to \alpha\in [0,1]$, where we can approximate $f(s)$ by the following integral
\begin{align}\label{f_def}
f(\alpha):=1+ \frac{1}{2} \int_{0}^1 \log(|t^2 - \alpha^2|) dt,
\end{align}
where for simplicity we used the same notation $f$ for the function with the normalized argument. With this approximation, we see that for $\frac{s}{M} \to \alpha \in [0,1]$, we have 
\begin{align}\label{f_a2m}
\lim_{M\to \infty} \frac{1}{2M} \log|a_{M}(s)|= \lim_{M \to \infty} h(s)=f(\alpha),
\end{align}
We can also evaluate the logarithmic integral in \eqref{f_def} as 
\begin{align}
\int_{0}^1 &\log(|t^2 - \alpha^2|)= \int_{0}^1 \log(t+\alpha) dt + \int_{0}^\varsigma \log(\alpha-t) dt\nonumber\\
& + \int_{\alpha}^1 \log(t-\alpha) dt\\
&=\int_{\alpha}^{1+\alpha} + \int_{0}^ \alpha + \int_{0}^{1-\alpha} \log(t) dt \\
&=\int_{0}^ {1+\alpha} + \int_{0}^{1-\alpha} \log(t) dt \\
&=(1+\alpha) \log(1+\alpha) - (1+\alpha) \\
&+ (1-\alpha) \log(1-\alpha) - (1-\alpha) \\
&=(1+\alpha) \log(1+\alpha)  + (1-\alpha) \log(1-\alpha) -2,
\end{align}
where we used the fact that the primitive function of $\log(t)$ is $t\log(t)-t$ for $t \in \bR_+$. This implies that 
\begin{align}
f(\alpha)&=\frac{1}{2} \Big ( (1+\alpha) \log(1+\alpha)  + (1-\alpha) \log(1-\alpha) \Big)\\
&=(1-h_2(\frac{1+\alpha}{2})) \log(2),
\end{align}
where $h_2(x)=-x\log_2(x)-(1-x)\log_2(1-x)$ is the binary entropy function for $x \in [0,1]$. 
Fig.\,\ref{fig:asymp_fig} illustrates $f(\alpha)$ for $\alpha\in[0,1]$ and its comparison with $\frac{1}{2M} |a_{2M}(s)|$ for $\frac{s}{M} \in [0,1]$ in the finite-$M$ regime for $M \in \{50, 100, 200\}$. It is seen that the approximation in quite tight even for $M=100$, where it is also seen that $f(\alpha)$ is an upper bound on $\frac{1}{2M} |a_{2M}(s)|$ for $\frac{s}{M} \in [0,1]$. It is also worthwhile to mention that $|a_{2M}(s)|=0$ for all the integers $s\in [M]$ since for an integer $s$ the even solution of Chebyshev ODE is a polynomial of order $2s$ with zero coefficients for terms with order higher than $2s$, thus, $\log(|a_{2M}(s)|) \to - \infty$ for all $\frac{s}{M} \in \{0,\frac{1}{M}, \dots, 1\}$, but $f(\alpha)$ is well-defined and continuous for all $\alpha \in [0,1]$. 
The following proposition summarizes some of the properties of $f(\alpha)$.
\begin{figure}[t]
	\centering
	\includegraphics[scale=0.95]{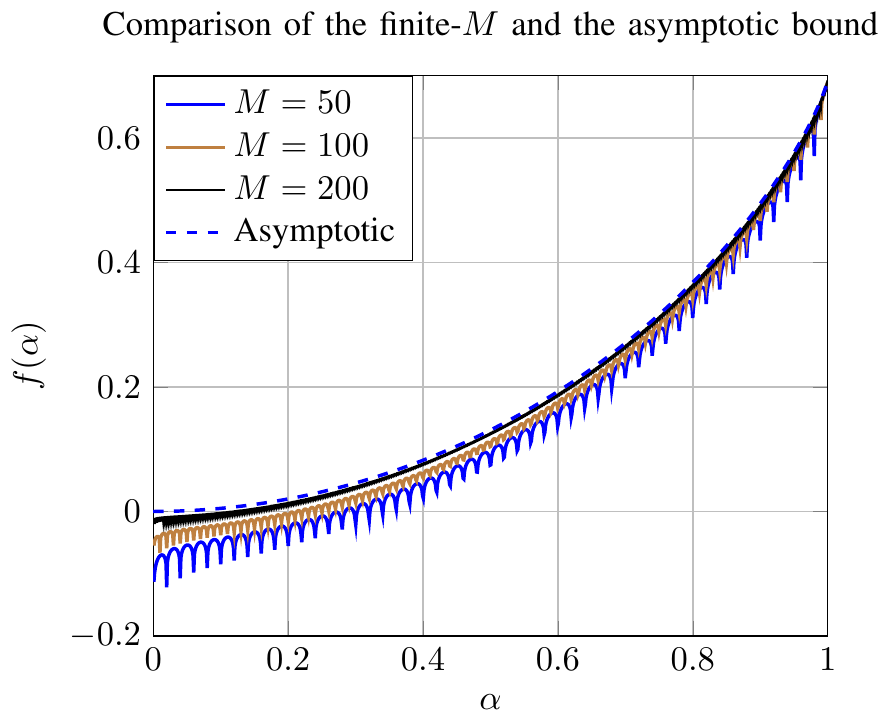}
\caption{Comparison between the asymptotic bound for $M\to \infty$ and the finite-$M$ bound for $M\in \{50, 100, 200\}$.}
\label{fig:asymp_fig}
\end{figure}
\begin{proposition}
Let $f(\alpha)$ for $\alpha \in [0,1]$ be as  before. Then, $f(\alpha)$ is a positive, convex, and increasing function of $\alpha$ for $\alpha \in[0,1]$. Moreover, $f(\alpha)\in [0, \log(2)]\approx [0, 0.6931]$. \hfill $\square$
\end{proposition}
\begin{proof}
	The proof simply follows from the properties of the binary entropy function $h_2$. The positivity follows from the fact that $h_2(\frac{1+\alpha}{2}) \in [0,1]$ for $\alpha \in [0,1]$. The increasing property follows from the fact that $x \mapsto h_2(x)$ is a decreasing function of $x$ for $x\in[\frac{1}{2},1]$, thus, $h_2(\frac{1+\alpha}{2})$ is a decreasing and  $f(\alpha)$ is an increasing function of $\alpha$ for $\alpha \in [0,1]$. The convexity also follows from the fact that $x \mapsto h_2(x)$ is a concave function of $x$ and $\alpha \mapsto \frac{1+\alpha}{2}$ is an affine function of $\alpha$. The last part also follows immediately from the monotonicity of $f$ and the fact that $h_2(\frac{1}{2})=1$ and $h_2(1)=0$. This completes the proof.
\end{proof}

\section{Bound on the Width of the Graph}\label{sec:minimax_analysis}
In this section, we will use the results obtained in Section \ref{sec:cheby} to derive upper bounds on the width of the graph $\clG$ associated to a fixed normalized measure $\gamma$ as introduced in Section \ref{sec:summary}. We will use this to prove Theorem \ref{main_thm}.
Recall that the definition of the graph $\clG$ in \eqref{graph} and its width in a fixed probing point  $s \in [0, M \varrho]$ in \eqref{graph_width}. 
\subsection{A Simple Minimax Bound}
Let $s \in [0,M\varrho]$ be a fixed probing point and define $\phi(\xi,s)=e^{j \pi s \xi}$ for $\xi \in [-1, 1]$. In this section, we consider the following simple problem. 
Suppose that $\xi\in [-1, 1]$ is fixed but unknown. The goal is to estimate $\phi(\xi, s)$ at an a priori known point $s$ from the following samples $\clE_\xi=\{\phi(\xi, k\varrho): k\in [M]\}$ (note the explicit dependence of $\clE_\xi$ on $\xi$). Recall that $\varrho$ is the spatial oversampling factor as introduced before. Let $\widehat{\phi}: \clE_\xi \to \bC$ be an estimator for $\phi(\xi, s)$ from the available samples $\clE_\xi$. We define the worst-case error of $\widehat{\phi}$ by
\begin{align}
e_{\widehat{\phi}}=\sup_{\xi \in [-1, 1]} |\phi(\xi, s)- \widehat{\phi}(\clE_\xi)|.
\end{align}
Here, we will mainly focus on linear estimators, where $\widehat{\phi}$ is a linear function of the observations $\clE_\xi$, where this linear function (linear estimator) can depend on the probing point $s$. We denote the set of all such linear estimators for a given $s$ by $\clL_s$. We define the minimax error over the class of linear estimators in $\clL_s$ by 
\begin{align}\label{ems_eq}
e_{M}(s)=\inf_{\widehat{\phi} \in \clL_s} \sup_{\xi \in [-1, 1]} |\phi(\xi,s)- \widehat{\phi}(\clE_\xi)|.
\end{align}
We prove the following result.
\begin{proposition}\label{width_prop}
	Let $s\in [0, M\varrho]$ be a fixed probing point. Let $\gamma$ be a fixed positive normalized measure  and let $\clG$ and  $w_\clG(s)$ be the  graph corresponding to $\gamma$ and its width at $s$. Let also $e_M(s)$ be as in \eqref{ems_eq}. Then, we have:
	
	\noindent {\bf 1.} $e_M(s) \leq 1$ for all $s$.
	
	\noindent {\bf 2.} $w_\clG(s) \leq 2 e_M(s)$. \hfill $\square$
\end{proposition}
\begin{proof}
The first part simply follows by setting $\widehat{\phi}=0$ to be the zero estimator (which is linear and belongs to $\clL_s$) and the fact that $|\phi(\xi,s)|=1$.

To prove the second part, note that by definition of $e_M(s)$, for any $\epsilon>0$, there is a linear estimator $\widehat{\phi}$ such that 
\begin{align}\label{dumm_min_2}
|\phi(\xi,s)- \widehat{\phi}(\clE_\xi)| \leq e_M(s)+ \epsilon,
\end{align}
 over the whole set $\xi \in [-1, 1]$. Consider  a positive normalized measure $\mu \in \Gamma_\gamma$ (see the definition of $\Gamma_\gamma$ in \eqref{Gamma_set}). Note that since the estimator $\widehat{\phi}$ is a linear function of $\clE_\xi$, we have that 
 \begin{align}
 \int \mu(d\xi) \widehat{\phi}&(\clE_\xi)= \widehat{\phi} \big(\int \mu(d\xi) \clE_\xi \big)\nonumber\\
 &=\widehat{\phi}\Big(\check{\mu}(0), \dots, \check{\mu}\big((M-1)\varrho\big)\Big)\nonumber\\
 &\stackrel{(a)}{=} \widehat{\phi}\Big(\gamf(0), \dots, \gamf\big((M-1)\varrho\big)\Big)=:c_\gamma(s),\label{dumm_min_3}
 \end{align}
 where $(a)$ follows from the fact that $\mu \in \Gamma_\gamma$, thus, it has the same Fourier transform as $\gamf$ at sampling points $\{k \varrho: k \in [M]\}$. Also, note that $c_\gamma(s)$ in the last expression  depends only on $\gamma,s$ but not on a specific $\mu \in \Gamma_\gamma$. From \eqref{dumm_min_2} and \eqref{dumm_min_3}, we obtain that 
 \begin{align}
 |\check{\mu}(s) - c_\gamma(s)|&=\Big| \int \mu(d\xi) \phi(\xi, s) - \int \mu(d\xi) \widehat{\phi}(\clE_\xi) \Big |\nonumber\\
 &=\Big| \int \mu(d\xi) \big (\phi(\xi, s) - \widehat{\phi}(\clE_\xi) \big ) \Big |\nonumber\\
 & \leq \int \mu(d\xi) \big |\phi(\xi, s) - \widehat{\phi}(\clE_\xi) \big |\\
 & \leq e_M(s) + \epsilon.
 \end{align}
 Since this is true for any $\epsilon>0$, we have that $ |\check{\mu}(s) - c_\gamma(s)| \leq  e_M(s)$. Note that this results is valid for any $\mu \in \Gamma_\gamma$, thus, by applying the triangle inequality, we have that $|\check{\mu}(s)-\check{\varepsilon}(s)| \leq 2e_M(s)$ for any arbitrary measure $\mu, \varepsilon \in \Gamma_\gamma$. From the definition of the width of $\clG$ in \eqref{graph_width}, this implies that $w_{\clG}(s) \leq 2e_M(s)$. This completes the proof.
\end{proof}

We will use Proposition \ref{width_prop} in the following to find an upper bound on $w_\clG(s)$ at any probing point $s$ by finding suitable upper bound for $e_M(s)$.

\subsection{Minimax error of the real part}
Let us first derive an upper bound on the minimax error of estimating the real part of $\phi(\xi,s)$ from the samples $\clE_\xi=\{\phi(\xi, k\varrho): k\in [M]\}$. This boils down in the linear minimax estimation of $\cos(\pi s \xi)$ from the real and the imaginary  parts of $\clE_\xi$ given by 
\begin{align}
\Re(\clE_\xi)&:=\Big \{1, \cos(\pi \varrho \xi), \dots, \cos\big (\pi \varrho (M-1)\xi\big)\Big \},\\
\Im(\clE_\xi)&:=\Big \{\sin(\pi \varrho \xi), \dots, \sin\big (\pi \varrho (M-1)\xi\big )\Big \}.
\end{align}
Since we are looking for an upper bound on the error, it is sufficient to consider only a subset of  linear estimators that use  only the real part $\Re(\clE_\xi)$.
 Denoting by $\bfc=(c_0, \dots, c_{M-1})^\transp\in \bR^M$ the coefficients of such a linear estimator, we can upper bound the minimax estimation error of the real part by 
\begin{align*}
e^{(r)}_M&(s)=\inf_{\bfc \in \bR^M} \max_{\xi \in [-1, 1]} |\cos(\pi s \xi) - \sum_{k=0}^{M-1} c_k \cos(k \pi \varrho \xi)|\\
&=\inf_{\bfc \in \bR^M} \max_{\xi^\circ \in [-\varrho, \varrho]} |\cos(\frac{\pi s}{\varrho} \xi^\circ) - \sum_{k=0}^{M-1} c_k \cos(k \pi  \xi^\circ)|\\
&=\inf_{\bfc}\hspace{-1mm} \max_{t \in [-\eta, \eta]} \Big |\cos(\frac{2s}{\varrho} \sin^{-1}(t)) -\hspace{-2mm} \sum_{k=0}^{M-1} c_k \cos(2k \sin^{-1}(t))\Big |,
\end{align*}
where the superscript $(r)$ refers to the real part, and where we made the change of variable $\xi^\circ=\varrho \xi$ and $t=\sin(\frac{\pi \xi^\circ}{2})$, and defined $\eta=\sin(\frac{\pi \varrho }{2})$, where $\eta \in (0,1)$ since $\varrho \in (0,1)$. Note that we have
\begin{align}
\sum_{k=0}^{M-1} c_k \cos(2k \sin^{-1}(t))&\stackrel{(a)}{=}\sum_{k=0}^{M-1} c_k \cos(k \pi + 2k \cos^{-1}(t))\nonumber\\
&=\sum_{k=0}^{M-1} c_k (-1)^k \cos(2k \cos^{-1}(t))\nonumber\\
&\stackrel{(b)}{=}\sum_{k=0}^{M-1} c_k (-1)^k T_{2k} (t) \label{cheb_expr_1}
\end{align}
where in $(a)$ we used the identity $\sin^{-1}(t)+ \cos^{-1}(t)=\frac{\pi}{2}$ for $t \in (-1,1)$, and where in $(b)$ we used the fact that for an integer $k$, $\cos(2k \cos^{-1}(t))$ coincides with the Chebyshev polynomial $T_{2k} (t) $ of order $2k$. Note that $\{T_{2k}(t): k\in [M]\}$ forms  a basis for the $M$-dim linear space of all even polynomials in $t\in (-1,1)$ of order at most $2(M-1)$. However, this space is also spanned  by  the monomials $\{t^{2k}: k \in [M]\}$. As a result, using \eqref{cheb_expr_1}, we can write the desired minimax error bound $e^{(r)}_M(s)$ more directly as
\begin{align*}
e^{(r)}_M(s)=\inf_{\bfb \in \bR^M} \max_{t \in [-\eta, \eta]} \Big |\cos(\frac{2s}{\varrho} \sin^{-1}(t)) - \sum_{k=0}^{M-1} b_{2k} t^{2k} \Big|,
\end{align*}
where we defined $\bfb=(b_0, b_2, \dots, b_{2(M-1)})^\transp\in \bR^M$. 
Replacing the coefficients $\bfb$ with the coefficients of Taylor's expansion of $\cos(\frac{2s}{\varrho} \sin^{-1}(t))$, i.e., $\{a_{2k}(\frac{s}{\varrho}): k \in [M]\}$ with the notation introduced in Section \ref{sec:cheby} (note that $\frac{s}{\varrho} \in [0,M]$), and using Proposition \ref{error_bound_prop}, we obtain the following upper bound 
\begin{align}
e^{(r)}_M(s) \leq |a_{2M}(\frac{s}{\varrho})| \frac{\eta^{2M}}{1-\eta^2} \approx \Big (\sin(\frac{\pi \varrho }{2}) g(\frac{s}{M\varrho})\Big )^{2M}
\end{align}
where in the last expression we replaced $\eta=\sin(\frac{\pi \varrho }{2})$ and used \eqref{f_a2m} and replaced $g =e^f$.

%
\subsection{Minimax error of the imaginary part}
We repeat similar steps to derive an upper bound on the minimax estimation of the imaginary part of $\phi(\xi,s)$, where this time we estimate the imaginary part of $\phi(\xi, s)$ from $\Im(\clE_\xi)$. More specifically, we consider the following minimax error  
\begin{align*}
e^{(i)}_M&(s)=\inf_{\bfs} \max_{\xi \in [-1, 1]} \Big |\sin(\pi s \xi) - \sum_{k=0}^{M-1} s_k \sin(k \pi \varrho \xi) \Big |\\
&=\inf_{\bfs} \max_{\xi^\circ \in [-\varrho, \varrho]} \Big |\sin(\frac{\pi s}{\varrho} \xi^\circ) - \sum_{k=0}^{M-1} s_k \sin(k \pi \xi^\circ) \Big |\\
&=\inf_{\bfs} \hspace{-1mm} \max_{t \in [-\eta, \eta]} \Big |\sin(\frac{2s}{\varrho} \sin^{-1}(t)) -\hspace{-2mm} \sum_{k=0}^{M-1} s_k \sin(2k \sin^{-1}(t)) \Big |,
\end{align*}
where the superscript $(i)$ refers to the imaginary part, where we defined $\bfs=(s_1, \dots, s_{M-1})^\transp \in \bR^{M-1}$,  and where we made the change of variable $\xi^\circ=\varrho \xi$ and $t=\sin(\frac{\pi \xi^\circ}{2})$ as before.
We obtain an upper bound on $e^{(i)}_M(s)$ via the estimation of the real part done before. 
We first define 
\begin{align*}
e_M^{(r)}(s,t):=\cos(\frac{2s}{\varrho} \sin^{-1}(t)) - \sum_{k=0}^{M-1} a_{2k} \cos(2k \sin^{-1}(t)),
\end{align*}
as the truncation error in the estimation of the real part, where $t \in [-\eta, \eta]$ as before, and where have used the same coefficients $a_{2k}=a_{2k}(\frac{s}{\varrho})$ as in the real case. Taking the derivative of $e_M^{(r)}(s,t)$ with respect to $t$ and some simplification yields
\begin{align}
&\max_{t \in [-\eta, \eta]} \Big | \frac{\varrho \sqrt{1-t^2}}{2s}  \frac{d}{dt} e_M^{(r)}(s,t)\Big |\nonumber\\
&:=\Big \| \sin(\frac{2s}{\varrho} \sin^{-1}(t) ) - \hspace{-2mm}\sum_{k=0}^{M-1} \frac{k}{s} a_{2k} \sin(2k \sin^{-1}(t)) \Big \|_\infty \label{exp_ss}\\
& \stackrel{(a)}{\geq} |e^{(i)}_M(s) |\label{dumm_imag_1}
\end{align}
where in $(a)$ we used the fact that \eqref{exp_ss} can be interpreted as the estimation error of an estimator with coefficients $s_k= \frac{k}{s} a_{2k}$, which can be lower bounded by $|e^{(i)}_M(s) |$ by definition. 
Applying Proposition \ref{error_bound_prop}, we can bound the derivative of the truncation error $e^{(r)}_{M}(s,t)$ in $t \in [-\eta, \eta]$ by
\begin{align*}
\max_{t \in [-\eta, \eta]} \big | \frac{d}{dt} e_M^{(r)}(s,t) \big | \leq 2|a_{2M}(\frac{s}{\varrho})| \frac{\eta^{2M-1} \big (M- (M-2) \eta \big)}{(1-\eta^2)^2}.
\end{align*}
From \eqref{dumm_imag_1}, this yields 
\begin{align}
e^{(i)}_M(s) &\leq \max _{t \in [-\eta, \eta]} |\frac{\varrho \sqrt{1-t^2}}{2s}| \max _{t \in [-\eta, \eta]} |\frac{d}{dt} e_M^{(r)}(s,t)|\\
&\leq  |a_{2M}(\frac{s}{\varrho})| \frac{\eta^{2M-1} (M- (M-2) \eta)}{s(1-\eta^2)^2}\\
&\stackrel{M\to \infty}{\approx} \Big (\sin(\frac{\pi \varrho }{2}) g(\frac{s}{M\varrho})\Big )^{2M}
\end{align}
where we used $\eta=\sin(\frac{\pi \varrho }{2})$ as before.

\subsection{Proof of Theorem \ref{main_thm}}
Combining the upper bound on the minimax error of the real part $e^{(r)}_{M}(s)$ and that of the imaginary part $e^{(i)}_{M}(s)$, we obtain an upper bound on the minimax error \eqref{ems_eq} as follows
\begin{align}
e_M(s) \leq e^{(r)}_{M}(s) + e^{(i)}_{M}(s) \leq C'  \Big (\sin(\frac{\pi \varrho }{2}) g(\frac{s}{M\varrho})\Big )^{2M}
\end{align}
where $C'$ is a universal constant independent of $M$ and $s$. From Proposition \ref{width_prop}, this yields the following upper bound on the width of the graph $\clG$ associated to a given measure $\gamma$ 
\begin{align*}
w_\clG(s) \leq \min \{ 2 e_M(s),2\} \leq \min \Big \{ C  \big(\sin(\frac{\pi \varrho }{2}) g(\frac{s}{M \varrho})\big )^{2M}, 2\Big \},
\end{align*}
 for some universal constant $C$ independent of $M$, $s \in [0, M\varrho]$, and the  measure $\gamma$. This proves Theorem \ref{main_thm}.

\section{Algorithms for UL-DL Covariance Interpolation}
In this part, we introduce two algorithms for UL-DL covariance interpolation by modifying two off-the shelf covariance estimation algorithms for massive MIMO.

\subsection{Interpolation from Instantaneous UL Channel Vectors}
For the first case, we assume that we have only access to the random realizations $\clH=\{\bfh_{\text{ul}}(t): t\in [T]\}$  of the UL channel vectors of a generic user inside an observation window consisting of $T$ time slots. We assume that the channel vectors are generated via a scattering geometry with an angular PSF $\gamma$ and have an UL  covariance matrix $\Sigmam_{\text{ul}}$, which is a PSD Toeplitz matrix with the  first column $\sigul$ given by \eqref{sig_dl_int_eq}.  
We assume that the samples $\clH$ are obtained by sampling sufficiently sparsely in time (separated by at least a coherence time) or in frequency (separated by more than coherence bandwidth inside the UL frequency band) such that they are i.i.d. Note that in the latter case, i.e., sparse sampling in the frequency domain, we always assume that the channel covariance matrix $\Sigmam(f)$, defined in  \eqref{cov_disc}, varies negligibly in the UL frequency band and can be well-approximated by  $\Sigmam_{\text{ul}}$.

We adopt here the low-complexity  covariance estimation algorithm proposed in \cite{haghighatshoar2016low, haghighatshoar2017low}. An interesting feature of this algorithm is that it directly provides an estimate of the underlying angular PSF $\gamma$. Here, for the sake of completeness, we briefly explain this algorithm. The algorithm requires a discretization of the $\aul$ over a discrete grid in the angular domain  of size $G \gg M$ given by $\clA=\{\theta_1, \dots, \theta_G\}\subset [-\theta_{\max}, \theta_{\max}]$, where we denote by
\begin{align}\label{A_grid_mat}
\bfA=[\aul(\theta_1), \dots, \aul(\theta_G)]\in \bC^{M \times G},
\end{align}
the $M\times G$ matrix consisting of the UL array responses $\aul(\theta)$ with $\theta \in \clA$. The algorithm in  \cite{haghighatshoar2016low, haghighatshoar2017low} requires only $m$-dim sketches of the UL channel vectors for some $m \ll M$ as in 
\begin{align}
\bfx(t)= \bfB(t) \bfh_{\text{ul}}(t) + \bfn(t), t \in [T],
\end{align}
where $\bfB(t) \in \bC^{m \times M}$ is a possibly time-variant projection matrix and where $\bfn(t) \sim \cg(0, \bfI_m)$ denotes the normalized noise power in the received sketches. Denoting by  $\check{\bfA}(t)=\bfB(t) \bfA$, $t\in [T]$, the projected UL channel responses at time slot $t$, the algorithm solves the following convex optimization problem for the $G \times T$ weighting matrix $\bfW=[\bfw(1), \dots, \bfw(T)]$
\begin{align}\label{f_cost}
f(\bfW)=\frac{1}{2}\sum_{i=1}^T \|\check{\bfA}(t) \bfw(t) - \bfx(t)\|^2 + \iota  \|\bfW\|_{2,1},
\end{align}
where $\iota \approx \sqrt{T}$ is a regularization parameter growing with the number of samples $T$, and where $ \|\bfW\|_{2,1}=\sum_{i=1}^G \|\bfW_{i,.}\|$ denotes the $l_{2,1}$ norm of the weighting matrix $\bfW$. We refer to \cite{haghighatshoar2016low, haghighatshoar2017low} for a low-complexity implementation of this algorithm and its extension to array configurations other than ULA. 

Let  $\bfW^*$ be the optimal minimizers of \eqref{f_cost}. The algorithm in \cite{haghighatshoar2016low, haghighatshoar2017low} proceeds by finding an approximation of the underlying angular PSF $\gamma$ as a discrete measure $\mu$ supported on  the AoA grid $\clA$ with 
\begin{align}\label{mes_estim}
\mu(\theta_i) \propto \|\bfW^*_{i,:}\|,  i \in [G].
\end{align} 
In \cite{haghighatshoar2016low, haghighatshoar2017low}, it was proved that if the AoA grid $\clA$ is sufficiently dense in $[-\theta_{\max}, \theta_{\max}]$ as $G \to \infty$ (typically $G =2M$ is enough for uniform grids), the estimate of the UL covariance matrix produced by $\mu$ (see e.g. \eqref{sig_f_eq}) given by
\begin{align*}
\widehat{\Sigmam}_\text{ul} = \int \mu(d\theta) \aul(\theta)\aul(\theta)^\herm=\sum_{i\in[G]} \mu(\theta_i) \aul(\theta_i)\aul(\theta_i)^\herm
\end{align*}
converges to the true UL covariance matrix. 

In view of our explanation in Section \ref{sec:main} and also Algorithm \ref{tab:UL_DL_interp}, we can use the estimate $\mu$ to interpolate the DL covariance matrix $\Sigdl$. More precisely, for a given spatial oversampling $\varrho$, 
we first obtain an estimate of the first column $\sigdl$ of $\Sigdl$ as in  Algorithm \ref{tab:UL_DL_interp}
\begin{align*}
\widehat{\sigmam}_\text{dl}=\int \mu(d\theta) \adl(\theta)=\sum_{i\in[G]} \mu(\theta_i) \adl(\theta_i)
\end{align*}
and keep only those coefficients of $\widehat{\sigmam}_\text{dl}$ with indices belonging to $\clI_\text{dl}(\varrho)$ as in \eqref{feas_index}  (see Fig.\,\ref{fig:graph_tradeoff} and the discussion in Section \ref{sec:trade}). This yields a robust interpolation of the DL covariance matrix from the observation of the UL channel vectors.

\subsection{Interpolation from  UL Covariance Matrix}\label{sec:nnls}
For the second case, we assume that one has a priori an estimate $\widehat{\Sigmam}_{\text{ul}}$ of the UL covariance matrix using other covariance estimation methods/algorithms.  Typically, in that case one obtains  an estimate of the covariance matrix directly rather than through the estimation of the underlying measure or a discretization thereof as in \eqref{mes_estim} proposed by the algorithm in \cite{haghighatshoar2016low, haghighatshoar2017low}. For the ULA, the covariance matrices are PSD and Toeplitz, thus, it is natural to assume that $\widehat{\Sigmam}_{\text{ul}}$ is also PSD and Toeplitz. Let us denote the first column of $\widehat{\Sigmam}_{\text{ul}}$ by $\widehat{\sigmam}_{\text{ul}}$. To do the covariance interpolation, we only need to find a positive measure $\mu$ that satisfies
\begin{align}\label{mu_estim_1}
\widehat{\sigmam}_{\text{ul}}=\int  \mu(d\theta) \aul(\theta),
\end{align}
as in \eqref{sig_dl_int_eq}. Using $\mu$ we immediately obtain an estimate of the first column of the DL covariance matrix $\sigdl$ as 
\begin{align}\label{sig_hat_dl_dumm}
\widehat{\sigmam}_{\text{dl}}=\int \mu(d\theta) \adl(\theta).
\end{align}
And, once we obtained $\widehat{\sigmam}_{\text{dl}}$, we only keep those initial components belonging to $\clI_\text{dl}(\varrho)$ as in \eqref{feas_index} that  are guaranteed to be robustly estimated  (see Algorithm \ref{tab:UL_DL_interp}, Fig.\,\ref{fig:graph_tradeoff} and the discussion in Section \ref{sec:trade}). It is important to note that even though there might be several such measures $\mu$ satisfying \eqref{mu_estim_1}, they are equally good for DL interpolation thanks to our performance guarantee in Theorem \ref{main_thm}. We can adopt several techniques to find a solution $\mu$ of the feasibility problem in \eqref{mu_estim_1}. One way is to discretize the set of AoAs into a sufficiently dense grid $\clA=\{\theta_1, \dots, \theta_G\}$ of size $G\gg M$  and look for a discrete approximation of the form
\begin{align}\label{mu_disc_estim}
\mu(d\theta) \approx \sum_{i=1}^G s_i \delta(\theta-\theta_i),
\end{align}
 for some $\bfs=(s_1, \dots, s_G)^\transp \in \bR_+^G$. The positive vector $\bfs$ can be found by the following simple convex optimization algorithm known as   \textit{non-negative least squares} (NNLS)
\begin{align}\label{nnls}
\widehat{\bfs}=\argmin_{\bfs \in \bR_+^G} \|\bfA \bfs - \widehat{\sigmam}_\text{ul}\|,
\end{align}
where $\bfA$ is the matrix consisting of the UL channel responses over $\clA$ as in \eqref{A_grid_mat}. 
In terms of numerical implementations, the NNLS in \eqref{nnls} can be posed as an unconstrained  Least-Squares problem over the positive orthant 
and can be solved  by  efficient techniques such as Gradient Projection, Primal-Dual techniques, etc., with an affordable computational complexity \cite{bertsekas2015convex}. We refer to \cite{kim2010tackling, nguyen2015anti} for the recent progress on the numerical solution of 
NNLS and  a discussion on other  related work in the literature. Using the estimate $\widehat{\bfs}$ in \eqref{nnls}, one can obtain an estimate of the measure $\mu$ as in \eqref{mu_disc_estim} and an estimate of $\widehat{\sigmam}_\text{dl}$ as in \eqref{sig_hat_dl_dumm} followed by an appropriate truncation.

\section{Simulation Results}
In this section, we illustrate the efficiency of our proposed UL-DL covariance interpolation via numerical simulations.
We assume that $\nu=\frac{f_\text{ul}}{f_\text{dl}}=0.9<1$ and that the ULA at the BS scans the angular range $\Theta=[-\theta_{\max}, \theta_{\max}]$ with a $\theta_{\max}=60$ degrees. 
We assume that $\gamma$ is a continuous distribution with a piece-wise constant density given by 
\begin{align}\label{rect_dist}
\gamma=\rect_{[0.6, 0.8]} + 4 \, \rect_{[0.8, 1]},
\end{align}
 where for $\clA \subset [-1,1]$, we denote by $\rect_{\clA}$ a rectangular pulse of amplitude $1$ in $\clA$ and $0$ elsewhere. Note that $\gamma$ is a normalized measure and  $\gamf(0)=\gamma([-1,1])=1$. 
For the simulations, we apply the NNLS interpolation algorithm introduced in Section \ref{sec:nnls} over a grid of size $G \gg M$ with an angular oversampling factor $\frac{G}{M}=4$, where $M$ denotes the number of antennas.

\subsection{Aliasing Effect (Grating Lobe) for $\varrho>1$}
We first consider the following simulation to illustrate the importance of the spatial oversampling factor $\varrho$. We  assume that the antenna spacing $d$ violates the conditions of  sampling theorem, that is, $\varrho>1$ and   $d=\frac{\varrho \lambda_\text{ul}}{2 \sin(\theta_{\max})}$ is larger than $\frac{\lambda_\text{ul}}{2 \sin(\theta_{\max})}$. In this case, even if $M \to \infty$, one might not be able to recover $\gamma$ uniquely due to the  aliasing, which is also known as the grating lobe effect in array processing literature \cite{van2002optimum}.  Fig.\,\ref{fig:aliasing} illustrates the UL-DL interpolation error at the DL sampling set. It is seen that even for  $\varrho=1.05$, which is only slightly larger than $1$, the UL-DL interpolation completely fails, thus, illustrating the importance of $\varrho<1$.  Also, note that, in this case, the interpolation error does not vanish by increasing the number of antennas $M$. 

\begin{figure}[t]
\centering
\includegraphics[scale=1]{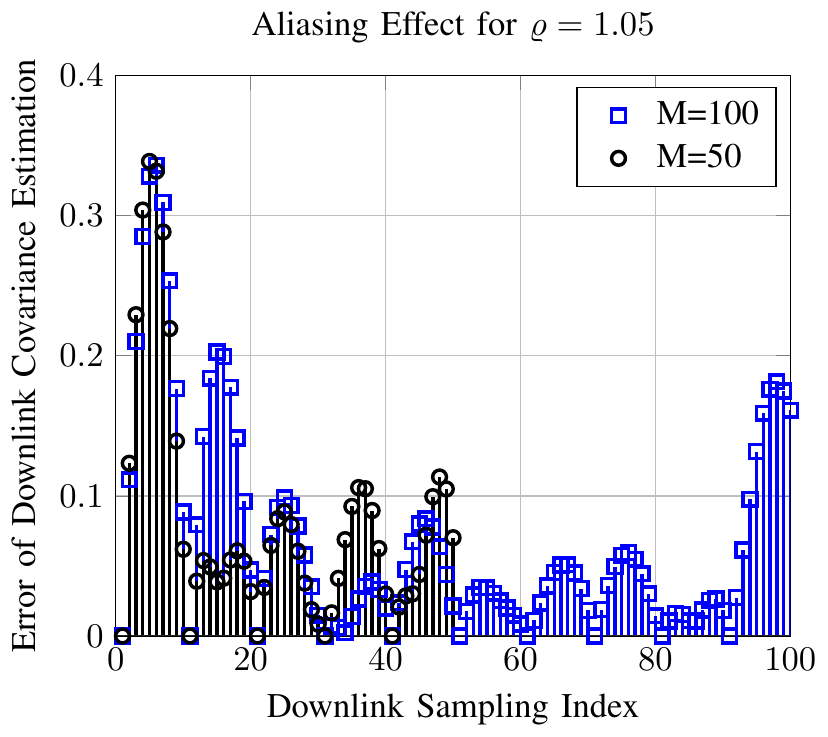}
\caption{UL-DL interpolation error on the DL sampling positions for different number of antennas $M \in \{50, 100\}$ when $\varrho=1.05>1$ and does not fulfill the sampling theorem. 
}
\label{fig:aliasing}
\end{figure}

\subsection{Simulation for $\varrho=0.9$}
We repeat the simulations, where this time we assume that $\varrho=0.9<1$ and fulfills the sampling theorem. Fig.\,\ref{fig:error_prop} illustrates the simulation results, where we again assume that the angular PSF $\gamma$ is as in \eqref{rect_dist}. It is seen that the proposed interpolation algorithm estimates $\gamf$ very well at all DL sampling locations, where the interpolation error grows quite fast  on the boundary locations, which consists of around $10\%$ of the samples. It is also seen that the error at the boundary points grows with increasing the number of antennas, which is compatible with Fig.\,\ref{fig:asymp_fig}, where the error exponent grows by increasing the number of antennas $M$ and approaches the positive error exponent $f(\alpha)$  as in \eqref{f_def} asymptotically as $M\to \infty$. 

\begin{figure}[t]
\centering
\includegraphics[scale=1]{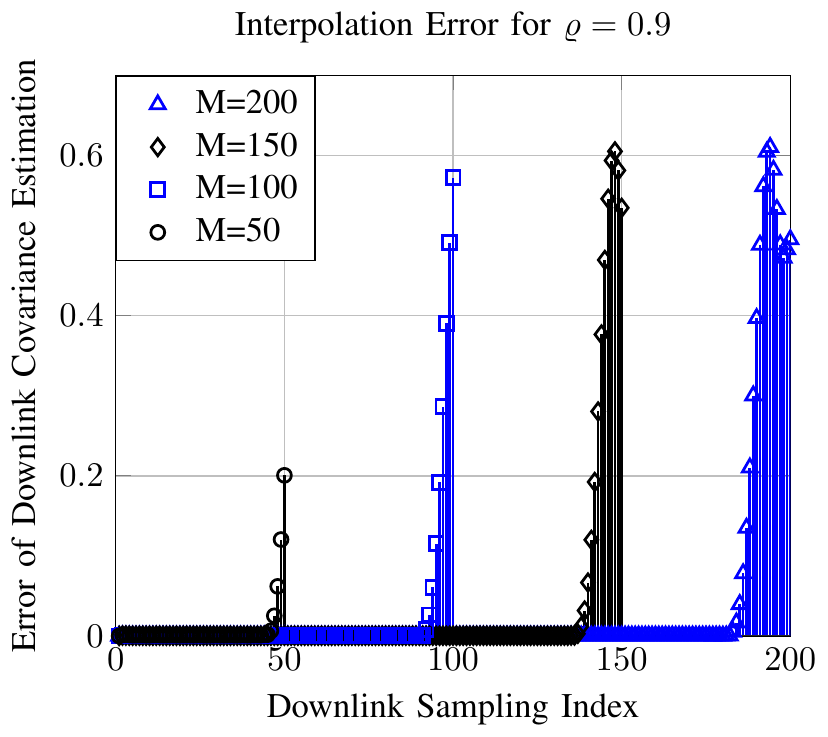}
\caption{UL-DL interpolation error on the DL sampling positions for different number of antennas $M$ when $\varrho=0.9<1$ and fulfills the sampling theorem. 
}
\label{fig:error_prop}
\end{figure}

It is also worthwhile to mention that the simulation results also illustrate that our proposed UL-DL interpolation algorithm seems to perform better than what predicted from Theorem \ref{main_thm}. This is partly due to the fact that the result of Theorem \ref{main_thm} has  a minimax nature and considers worst-case angular PSFs $\gamma$. It would be interesting to sharpen the result of Theorem \ref{main_thm} for a more structured $\gamma$ via a more refined analysis. We leave this is as an important problem for the future work.

\subsection{Performance of UL-DL Covariance Interpolation}
As we explained in Algorithm \ref{tab:UL_DL_interp}, after obtaining an estimate of $\sigdl$ from the available $\sigul$, we need to truncate the last elements of the resulting estimate $\sigdl$ not belonging to the index set $\clI_\text{dl}(\varrho)$ to get rid of the interpolation error on the boundary (see also Fig.\,\ref{fig:graph} and Fig.\,\ref{fig:graph_tradeoff}). 
Our goal in this section is to investigate how this truncation affects the system performance in a massive MIMO setup. 

Let $\widehat{\Sigmam}_\text{dl}$ be the $M \times M$ Toeplitz matrix that contains the truncated estimate $\sigdl$ as its first column. 
Let  $\Sigmam=\bfU \Lambdam \bfU^\herm$ and $\widehat{\Sigmam}_\text{dl}=\widehat{\bfU} \widehat{\Lambdam} \widehat{\bfU}^\herm$ denote the singular value decompositions (SVDs) of the true DL covariance matrix  $\Sigmam_\text{dl}$ and its estimate $\widehat{\Sigmam}_\text{dl}$, where  $\Lambdam$ and $\widehat{\Lambdam}$ are the diagonal matrices of singular values $\lambdam=(\lambda_1, \dots, \lambda_M)^\transp$ and $\widehat{\lambdam}=(\widehat{\lambda}_1, \dots, \widehat{\lambda}_M)^\transp$.  We always use the convention that the singular values are sorted in a non-increasing order. We define the normalized power distribution for $\Sigdl$ by $\bfp=(p_1, \dots, p_M)^\transp \in \bR_+^M$, where $p_i=\frac{\lambda_i}{\sum_{j=1}^M \lambda_j}$ denotes the fraction of the power of the DL channel vector $\bfh_\text{dl}(t)$ that lies in the rank-1 subspace $\bfu_i \bfu_i^\herm$, where $\bfu_i$ denotes the $i$-th column of $\bfU$. Let $\widehat{\bfU}=[\widehat{\bfu}_1, \dots, \widehat{\bfu}_M]$ where $\widehat{\bfu}_i$ denotes the $i$-th column of $\widehat{\bfU}$. We denote the power captured by columns of $\widehat{\bfU}$ by $\bfq=(q_1, \dots, q_M)^\transp \in \bR_+^M$, where $q_i=\inp{\Sigdl}{\widehat{\bfu}_i \widehat{\bfu}_i^\herm}$ gives the amount of power of $\Sigdl$ captured by the rank-1 subspace  $\widehat{\bfu}_i \widehat{\bfu}_i^\herm$ of the estimate $\widehat{\Sigmam}_\text{dl}$. It is not difficult to check that $\sum_{i=1}^M q_i=\tr(\Sigdl)$, which gives the whole power contained in $\Sigdl$. We normalize $\bfq$ and define the estimated normalized power distribution $\widehat{\bfp}=(\widehat{p}_1, \dots, \widehat{p}_M)^\transp\in \bR_+^M$, where $\widehat{p}_i=\frac{q_i}{\sum_{j=1}^M q_j}$. 
Let $\eta_\bfp(k):=\sum_{i=1}^k p_i$ and $\eta_{\widehat{\bfp}}(k)=\sum_{i=1}^k \widehat{p}_i$, for $k\in[M]$, denote the whole signal power contained in the first $k$ component of $\bfp$ and $\widehat{\bfp}$. Note that since $\bfU$ is the SVD basis for $\Sigdl$, we always have $\eta_{\bfp}(k) \geq \eta_{\widehat{\bfp}}(k)$, for every $k \in [M]$, that is,  $\widehat{\bfp}$ is always majorized by ${\bfp}$, where also, due to the normalization,  $\eta_{\bfp}(M)=\eta_{\widehat{\bfp}}(M)=1$.

In almost all applications of massive MIMO in which the covariance information is used to improve the system performance, one in interested in one way or another to find a subspace of sufficiently large dimension $k_0 \in [M]$ that captures a significant fraction of the power of user channel vectors\footnote{Namely, an $M\times k_0$ matrix $\bfV$, for some $k_0\ll M$, with $\bfV^\herm \bfV=\bfI_{k_0}$, and $\bE\big[\|\bfV^\herm \bfh_\text{dl}(t)\|^2\big] \geq (1-\epsilon) \bE\big[\|\bfh_\text{dl}(t)\|^2\big]$, for some small $\epsilon \in (0,1)$.} (e.g., $\bfh_\text{dl}(t)$ in a DL scenario). This can be posed as finding the smallest $k_0$ with $\eta_{\bfp}(k_0) \geq 1-\epsilon$ for a sufficiently small $\epsilon \in (0,1)$. 
In practice, however,  the subspace estimation can be done only through an estimate of the covariance matrix (e.g., $\widehat{\Sigmam}_\text{dl}$ in our case), and one ends up obtaining only a fraction $\eta_{\widehat{\bfp}}(k_0)$ of the power of channel vectors. Thus, a \textit{good} distortion measure as in \cite{haghighatshoar2017low, haghighatshoar2016low} is $\vartheta(\bfp, \widehat{\bfp})=\max_{k \in [M]} \frac{\eta_{{\bfp}}(k) - \eta_{\widehat{\bfp}} (k)}{\eta_{{\bfp}}(k)}$, which takes this power loss into account for any arbitrary signal dimension $k \in [M]$. Note  that $\vartheta(\bfp, \widehat{\bfp})\in [0,1]$, and $\vartheta(\bfp, \widehat{\bfp})=0$ if and only if $ \widehat{\Sigmam}_\text{dl}=\beta \Sigdl$ for some $\beta >0$. 

We repeat our simulation with the same $\gamma$ as in \eqref{rect_dist}. We consider three simulation scenarios:
\begin{enumerate}
\item We perform no  UL-DL interpolation and use an estimate of the UL covariance matrix $\Sigul$ for the DL channel to evaluate the effect of frequency-dependent  distortion.
\item We apply UL-DL interpolation to estimate $\Sigdl$ from an estimate of $\Sigul$ but we do not apply any truncation at the boundary values (see Fig.\,\ref{fig:graph_tradeoff} and Fig.\,\ref{fig:graph}).

\item We do  UL-DL interpolation followed by  truncation of  boundary values ($10\%$ truncation) as in Algorithm \ref{tab:UL_DL_interp}.
\end{enumerate}
Fig.\,\ref{fig:dist_diff_M} illustrates the simulation results. It is seen that, without compensating the frequency-dependent covariance distortion via UL-DL interpolation, the distortion measure $\vartheta(\bfp, \widehat{\bfp})$ is significantly large even for $M=25$ antennas and grows by increasing $M$. In contrast, UL-DL interpolation dramatically reduces the distortion $\vartheta(\bfp, \widehat{\bfp})$ but  it shows its excellent  performance when it is accompanied by truncation. 
It is also interesting to note that truncation degrades the performance for small number of antennas $M$ but tends to improve it significantly when  $M$ is moderately large ($M \gtrsim 100$).

\begin{figure}[t]
	\centering
	\includegraphics[scale=1]{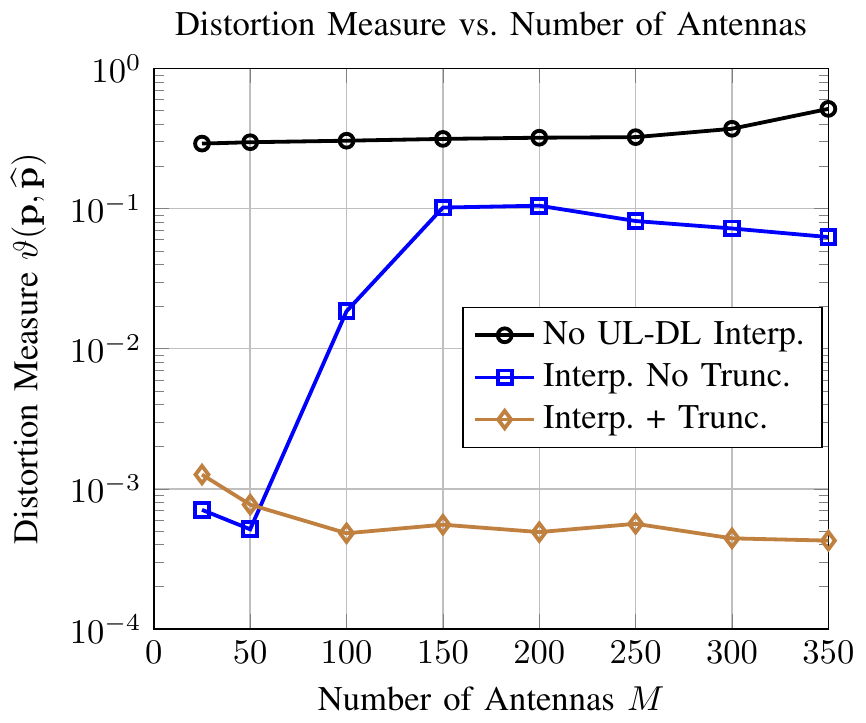}
	\caption{Distortion measure $\vartheta(\bfp, \widehat{\bfp})$ vs. number of antennas $M$ for different DL covariance estimation scenarios: 1)\,No UL-DL interpolation, 2)\,UL-DL interpolation without any truncation, 3)\,UL-DL interpolation followed by $10\%$ truncation.}
	\label{fig:dist_diff_M}
\end{figure}

\section{Extension to other Array Geometries}
In this paper, we studied  UL-DL covariance interpolation when the BS has a ULA with $M$ antennas. Recently, there is a huge interest to use higher-dim (full-dim) antenna arrays  for massive MIMO. In this section, we explain briefly how one can extend the results in this paper to those scenarios.
For higher-dim array geometries, one  can generally parametrizes the AoAs by points over the unit sphere $\bS:=\{\xim \in \bR^3: \|\xim\|=1\}$ in $\bR^3$ and represent the angular PSF (power profile) of each user as a positive measure $\gamma$ over $\bS$. Denoting by $\clR:=\{\bfr_k \in \bR^3: k \in [M]\}$ the position of antenna elements in the BS array, we can define UL/DL array responses similar to \eqref{both_arr_resp} as $\aul(\xim)$ and $\adl(\xim)$, where 
\begin{align}
[\aul(\xim)]_k=e^{j k \frac{2\pi}{\lambda_{\text{ul}}} \inp{\xim}{\bfr_k}}, [\adl(\xim)]_k=e^{j k \frac{2\pi}{\lambda_{\text{dl}}} \inp{\xim}{\bfr_k}},
\end{align}
for $k \in [M]$. Similarly, we can introduce the Fourier transform of the angular PSF $\gamma$ as
\begin{align}
\gamf(\bfr)=\int _{\bS}  e^{j \pi \inp{\xim}{\bfr}} \gamma(d\xim),
\end{align}
where the integral is taken over the set of all AoAs parameterized with $\xim \in \bS$. 
Using \eqref{sig_f_eq} and following similar steps, we obtain that
\begin{align}
[\Sigul]_{k,l}&= \int _{\bS}  e^{j \pi \inp{\xim}{\frac{\bfr_k-\bfr_l}{\frac{\lambda_{\text{ul}}}{2}}}} \gamma(d\xim)=\gamf(\frac{\bfr_k-\bfr_l}{\frac{\lambda_\text{ul}}{2}}),\\
[\Sigdl]_{k,l}&= \int _{\bS}  e^{j \pi \inp{\xim}{\frac{\bfr_k-\bfr_l}{\frac{\lambda_{\text{dl}}}{2}}}} \gamma(d\xim)=\gamf(\frac{\bfr_k-\bfr_l}{\frac{\lambda_\text{dl}}{2}}).
\end{align}
Denoting by $\clD:= \clR-\clR=\{\bfr_k - \bfr_l: k,l\in [M]\}$ the Minkowski difference \cite{tao2006additive} of the antenna geometry $\clR:=\{\bfr_k: k \in [M]\}$, we can see that the UL-DL covariance interpolation in this setup can be posed, similar to that stated in Problem \ref{main_prob},  as the problem of recovering  the Fourier transform $\gamf$ of $\gamma$ at the DL sampling set $\clD_\text{dl}:=\frac{\clD}{\frac{\lambda_\text{dl}}{2}}$ from  its value at UL sampling set  $\clD_\text{ul}:=\frac{\clD}{\frac{\lambda_\text{ul}}{2}}$, where for an $\alpha \in \bR$ we denote by $\alpha \clD:=\{\alpha \phim: \phim \in \clD\}$ the component-wise scaling of the elements of $\clD$ by the factor $\alpha$. 
Interestingly, it is seen that, as in the 1-dim case of ULA, the set of DL sampling positions is simply given by  $\clD_\text{dl}= \frac{1}{\nu} \times \clD_\text{ul}$, which is a scaled version of the set of UL sampling positions $\clD_\text{ul}$ by a factor $\frac{1}{\nu}$ (larger than $1$ for $\nu<1$),  where  $\nu=\frac{\lambda_{\text{dl}}}{\lambda_{\text{ul}}}=\frac{f_{\text{ul}}}{f_{\text{dl}}}$, denotes the ratio between the UL and DL carrier frequencies.
This is illustrated for a circular array geometry in Fig.\,\ref{fig:circ_grid}, where it is seen that $\clD_\text{ul}$ and $\clD_\text{dl}$ each consist of $O(M^2)$  points.  

For the ULA studied in this paper, the underlying geometry is 1-dim and consists of $\clR=\{k d\, \xim_0: k \in [M]\}$ for some unit vector $\xim_0 \in \bS$ and some antenna spacing $d$. Thus, the difference set $\clD=\clR-\clR=\{k d\, \xim_0: k=-(M-1),\dots, (M-1)\}$ is also 1-dim, lies along $\xim_0$, and consists of $2M-1$  points (rather than $O(M^2)$  points as in a circular array). Also, for the ULA, the resulting UL/DL covariance matrices are Toeplitz matrices that depend on a 1-dim normalized measure that is obtained by projecting $\gamma(d\xim)$ onto the collection of 2-dim planes that are orthogonal to $\xim_0$, and can be represented by $\gamma(d \xi)$ in terms of the parameter $\xi:=\inp{\xim_0}{\xim}$ that takes values in $[-1,1]$ as $\xim$ varies over the  sphere $\bS$. Thus, the UL-DL covariance interpolation problem boils down to that stated in Problem \ref{main_prob}.

We  expect that when $\nu<1$, thus, the DL sampling set $\clD_\text{dl}$ is an expended version of UL one $\clD_\text{ul}$, as in the 1-dim case of ULA, one needs to impose some spatial oversampling by a  factor $\varrho <1$ in our notation  to guarantee a stable UL-DL covariance interpolation, namely, the array elements should be closely spaced (measured in terms of $\frac{\lambda_\text{ul}}{2}$). 
Note that in the 1-dim case, we had to apply some additional truncation on the estimated DL samples whose positions  were close to the boundary of the DL sampling window  $[0,M \varrho]$ since these samples, intuitively speaking, might not be reliably estimated. We  illustrated, via numerical simulations (see, e.g., Fig.\,\ref{fig:dist_diff_M}), that this additional truncation  further improves the performance. 
A similar situation holds for higher-dim arrays.
In particular, for any array geometry $\clR$, the Minkowski difference  $\clD=\clR-\clR$ is a symmetric set, i.e., $\clD=-\clD$, centered at the origin ${\bf 0}$, and $\clD_\text{ul}$ and $\clD_\text{dl}$ correspond to the directional scaling of this set with respect to the origin by a factor $\frac{2}{\lambda_\text{ul}}$ and $\frac{2}{\lambda_\text{dl}}$ respectively. As a result, one can always identify a well-defined boundary between $\clD_\text{ul}$ and $\clD_\text{dl}$. 
For example, for a circular array illustrated in Fig.\,\ref{fig:circ_grid}, this boundary corresponds to all DL sampling points (solid squares) that lie outside  the boundary circle corresponding to the UL sampling points (solid circles).
One can apply truncation at those boundary points to further improve the performance of the interpolation.
In this paper, we used the properties of the series solutions of  Chebyshev ordinary differential equation (since we had a 1-dim variable $\xi$) to derive bounds on the required spatial oversampling factor $\varrho$ and to specify the subset of reliable samples as in \eqref{feas_index} in a minimax setup. It would be interesting to derive similar bounds using perhaps tools from  partial differential equations (since $\xim$ is 2-dim). We leave this as an interesting problem to be further investigated in a future work.

\begin{figure}[t]
\centering
\begin{subfigure}[b]{0.15\textwidth}
        \centering
        \includegraphics[scale=1]{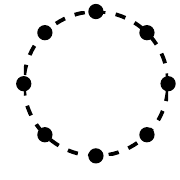}
        \caption{Circular array.}
\end{subfigure}%
\begin{subfigure}[b]{0.3\textwidth}
        \centering
        \includegraphics[scale=1]{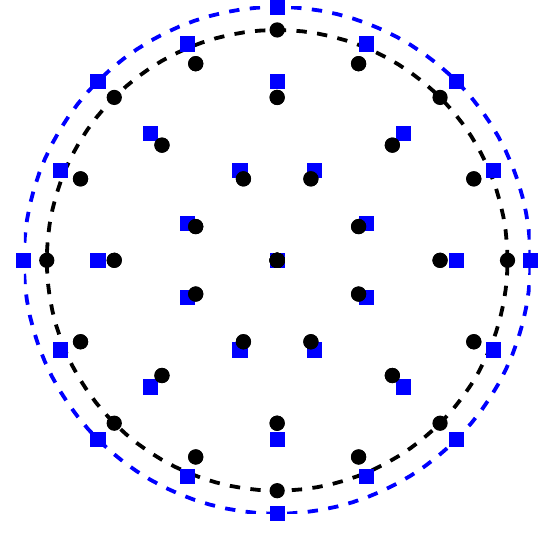}
        \caption{Corresponding UL/DL sampling sets.}
\end{subfigure}%
    
\caption{Illustration of a circular array geometry $\clR$ and the corresponding UL (solid circle) and DL (solid square) sampling positions $\clD_\text{ul}= \frac{\clR-\clR}{\frac{\lambda_\text{ul}}{2}}$ and $\clD_\text{dl}=\frac{1}{\nu} \times \clD_\text{ul}$.}
\label{fig:circ_grid}
\end{figure}

\section{Conclusion}
In this paper, we studied the effect of frequency-dependent distortion of user channel covariance matrix in a massive MIMO setup. This problem arises because of the variation of the array response of the BS antennas with frequency. We explained that although this effect is generally negligible for a small number of antennas $M$, it results in a considerable distortion of the covariance matrix in  the massive MIMO regime where $M \to \infty$.
We proposed some mild conditions, namely, stationarity of the channel process and the reciprocity of angular \textit{power spread function} of each user between the UL and the DL, under which the  covariance matrix in the DL can be suitably interpolated  from  that in the UL.  We  analyzed the interpolation problem mathematically and proved its robustness under a sufficiently large spatial oversampling of the array. We also proposed a simple scheme, consisting in using \textit{off-the-shelf} algorithms for interpolation followed by truncating the unreliably estimated values at the boundary of DL sampling positions,  and evaluated its performance in a massive MIMO setup via numerical simulations. Our simulation results clearly indicate that the frequency-dependent distortion incurs a considerable degradation of the performance, which is  compensated almost perfectly by applying our proposed UL-DL interpolation algorithm. Finally, we proposed guidelines for  extending our results to general, especially higher-dim, array geometries rather than the 1-dim case of ULA studied in this paper.


\balance

{\small
\bibliographystyle{IEEEtran}
\bibliography{references2}
}

\end{document}